\documentclass[a4,11pt]{elsarticle}
\usepackage{mathtools}
\usepackage{refcount}
\usepackage[utf8]{inputenc}
\usepackage[colorinlistoftodos]{todonotes}
\usepackage[bookmarks,colorlinks,breaklinks]{hyperref}
\usepackage[breakable]{tcolorbox}
\hypersetup{urlcolor=blue, colorlinks=true, citecolor=green!50!black,linkcolor=blue}
\usepackage{enumerate}
\usepackage{tikz} 
\usetikzlibrary{snakes}
\usetikzlibrary{shapes,shapes.geometric,arrows,fit,petri,calc,positioning,automata,decorations.pathmorphing}
\usetikzlibrary{shapes.geometric}
\tikzset{state/.style={draw,ellipse}}
\newif\ifextended
\extendedfalse

\usepackage[ruled]{algorithm}
\usepackage[noend]{algpseudocode}

\usepackage{qtree}

\usepackage{xcolor}
\usepackage{amsmath}
\usepackage{amsthm}
\usepackage{amssymb}
\usepackage{amsfonts}

\usepackage{bbding}
\newcommand{\Hand}{\raisebox{-.1cm}{\HandRight}}
\usepackage{scalerel}
\makeatletter
\newcommand{\hourglass}{\mathbin{\scalerel*{\@hgpic}{\ensuremath{\sigma}}}}\newcommand{\@hgpic}{\setlength{\unitlength}{0.34cm}\begin{picture}(1,1.5)\thicklines \put(0,0){\line(2,3){1}}\put(1,1.5){\line(-1,0){1}}\put(0,1.5){\line(2,-3){1}}\put(1,0){\line(-1,0){1}}\end{picture}}
\makeatother
\newcommand{\wait}{\$}

\newcommand{\ZO}{\{0, 1\}}

\newcommand{\eps}{\varepsilon}

\newcommand{\Step}{\mathsf{Step}}
\newcommand{\NT}{\mathsf{NT}}
\newcommand{\rec}{\mathsf{Rec}}

\newcommand{\PEG}{\mathsf{PEG}}

\newcommand{\gnot}{\text{\tt !}}
\newcommand{\gand}{\text{\tt \&}}

\newcommand{\HasState}{{\mathsf{State}}}
\newcommand{\HasLabel}{{\mathsf{Label}}}

\newcommand{\Path}{{\mathsf{Path}}}
\newcommand{\HasNeighbourhood}{{\mathsf{Neighbourhood}}}
\newcommand{\AutomatonAccepts}{{\mathsf{AutomatonAccepts}}}
\newcommand{\Transition}{{\mathsf{Transition}}}
\newcommand{\ACCEPT}{{\eps}}
\newcommand{\FAIL}{{\mathsf{FAIL}}}

\newcommand{\cA}{{\mathcal A}}

\newcommand{\cE}{{\mathcal E}}

\newcommand{\bbN}{{\mathbb N}}

\newcommand{\cG}{{\mathcal G}}

\newcommand{\cH}{{\mathcal H}}

\newcommand{\cL}{{\mathcal L}}

\newcommand{\cM}{{\mathcal M}}

\newcommand{\cN}{{\mathcal N}}

\newcommand{\cP}{{\mathcal P}}

\newcommand{\bbS}{{\mathbb S}}

\newcommand{\NC}{\mathsf{NC}}

\newcommand{\PTIME}{\mathsf{P}}

\newcommand{\bsni}{\bigskip\noindent}

\newcommand{\Online}{\mathsf{Online}}
\newcommand{\BinDepth}{\mathsf{BinDepth}}
\newcommand{\BinTree}{\mathsf{BinTree}}

\theoremstyle{definition}
\newtheorem{theorem}{Theorem}

\newtheorem{definition}[theorem]{Definition}

\newtheorem{corollary}[theorem]{Corollary}
\newtheorem{claim}[theorem]{Claim}
\newtheorem{conjecture}[theorem]{Conjecture}

\newtheorem{invariant}{Invariant}

\newcommand{\ps}[1]{^{(#1)}}

\newcommand{\thankstext}{Bruno Loff is the recipient of FCT postdoc grant number SFRH/\allowbreak BPD/\allowbreak 116010/\allowbreak 2016. This work is partially funded by the ERDF through the COMPETE 2020 Programme  within project POCI-01-0145-FEDER-006961, and by National Funds through the FCT as part of project  UID\allowbreak /EEA\allowbreak /50014\allowbreak /2013. This work was partially supported by CMUP (UID\allowbreak /MAT\allowbreak /00144\allowbreak /2019), FCT (Portugal), FEDER and PT2020. The authors would like to thank Markus Holzer, Martin Kutrib, Leen Torenvliet and Jurgen Vinju for fruitful discussions on this subject.}

\title{The computational power of\\ Parsing Expression Grammars\footnote{This is a revised and expanded version of a paper presented at the 22nd International Conference on Developments in Language Theory (DLT), held in Tokyo, Japan, September 10-14, 2018.}}

\begin{document}
\author[1,2]{Bruno Loff\corref{cor1}}
\ead{bruno.loff@gmail.com}
\author[1,3]{Nelma Moreira}
\ead{nam@dcc.fc.up.pt}
\author[1,3]{Rogério Reis}
\ead{rvr@dcc.fc.up.pt}
\address[1]{DCC, Faculdade de Ciências da Universidade do Porto}
\address[2]{CRACS, INESC-Tec}
\address[3]{CMUP, Universidade do Porto}
\cortext[cor1]{Corresponding author}
\date{\today}

\begin{abstract}
  We study the computational power of parsing expression grammars (PEGs). We begin by constructing PEGs with unexpected behaviour, and surprising new examples of languages with PEGs, including the language of palindromes whose length is a power of two, and a binary-counting language.
  
  We then propose a new computational model, the \emph{scaffolding automaton}, and prove that it exactly characterises the computational power of parsing expression grammars (PEGs).

  Several consequences will follow from this characterisation: (1) we show that  PEGs are computationally ``universal'', in a certain sense, which implies the existence of a PEG for a P-complete language; (2) we show that there can be no pumping lemma for PEGs; and (3) we show that PEGs are strictly more powerful than online Turing machines which do $o(n/(\log n)^2)$ steps of computation per input symbol.
\end{abstract}


\maketitle

\newpage

\tableofcontents

\section{Introduction}

Parsing expression grammars are a recognition-based system for parsing of formal languages. They were defined by Ford \cite{ford2004parsing}, who showed equivalence with earlier parsing systems by Birman and Ullman \cite{birman1970parsing,birman1970tmg} that are able to recognise the class of  \emph{top-down parsing languages} (TDPLs \cite{aho1972theory}).

As a language formalism, PEGs offer an attractive syntax and an efficient linear-time parsing algorithm which is nonetheless simple to implement. This led to a recent trend, which pushes for the adoption of PEGs, both as a theoretical subject \cite{chida17:_linear_parsin_expres_gramm,garnock-jones18:_recog_gener_terms_deriv_parsin_expres_gramm,henglein17:_peg,mizushima10:_packr,Moss:2017ux,Redziejowski:2013bw,Redziejowski:bc,redziejowski18:_tryin_under_peg}, and as a practical tool for parser generators \cite{Becket:2008vb,grimm2006better,Ierusalimschy:2009hl,koprowski11:_trx,kuramitsu16:_fast_flexib_declar_const_abstr,laurent15:_parsin,maidl16:_error_parsin_expres_gramm,Matsumura:2015uu,Medeiros:2008tya}.
See Ford's webpage \cite{fordwebpage} for an extensive bibliography of work around PEGs.

The influence of PEGs is illustrated by the surprising fact that, despite having been introduced only fifteen years ago, the number of available PEG-based parser generators already seems to nearly-match or even supersede the number of parser generators based on any other single parsing method, even when compared with methods which are many decades older.\footnote{We estimate this to be true, based on consulting the Wikipedia page ``Comparison of parser generators'', and searching GitHub for ``parser generator X'', and then counting how many projects appear which use a given method X. Doing so, one obtains the following numbers (ca. September 2019):
  \begin{center}
    \begin{tabular}{r|c|c|c|c|c|c}
                & LR & LL & LALR & GLR & Earley & \textbf{PEG}\\
Wikipedia & 26 & 33 & 63   & 23  & 7      & \textbf{48}\\
GitHub    & 62 & 86 & 77   & 10   & 9      & \textbf{122}
    \end{tabular}
\end{center}} This seems to be due to the simplicity of the formalism, which allows for the quick appearance of many small DIY projects; the situation is reversed if limits one's attention to high-quality projects, and there does not yet appear to be any serious global tendency to replace older technologies by PEGs. Nonetheless, a few high-quality PEG-based parser generators do exist (e.g. \emph{rats!} \cite{grimm2006better}, or the \emph{Scala Standard Parser-Combinator Library}), and there was at least one serious, influential attempt at creating a programming language which intrinsically relied on PEG as a parsing technology --- the \emph{Fortress} programming language \cite{steele1999growing}, which was being developed by Guy Steele's team at Sun Microsystems. The project is now defunct, but Fortress was once considered as a possible \emph{next-generation} replacement for the \emph{Java} programming language \cite{flood2008fortress}! 

Despite this enthusiasm for PEGs, we have also started seeing some objections of a theoretical nature. On one hand, proving the correctness of a given parsing expression grammar is often more difficult than one would like, even for simple examples\footnote{For example, the relatively simple grammar for the $a^n b^n c^n$ language which appears in Ford's original paper \cite{ford2004parsing}, has a (fixable) bug, which eluded discovery for over a decade (including to us, when we read Ford's paper) until the bug was pointed out by a recent paper of Garnock-Jones et al.~\cite{garnock-jones18:_recog_gener_terms_deriv_parsin_expres_gramm}.}. This makes PEGs somewhat problematic as a model of formal languages. On the other hand, there is no natural example of a language which is proven not to have PEGs. We believe that the present work will help in understanding why this is the case.

A first naive look at PEGs may suggest that their computational power should be roughly similar to that of deterministic context-free grammars \cite{ford2004parsing}. Indeed it is known that deterministic context-free languages have PEGs \cite{birman1970parsing}. But already Aho and Ullman \cite{aho1972theory} had shown that the $a^n b^n c^n$ language, which is not context-free, is still a TDPL, and hence has a PEG \cite{ford2004parsing}.

One may still hope that the computational power of PEGs can be contained, in some way, akin to how we can use pumping lemmas to separate the Chomsky hierarchy (e.g. \cite{bar1964formal,li1995new,Hayashi:aa,yu1989pumping,amarilli2012proof}). The following question appears in Aho and Ullman's book \cite{aho1972theory}, and in Ford's article \cite{ford2004parsing}:

\begin{center}
  \em    Is there a context-free language without a parsing expression grammar?
\end{center}
\noindent
It is possible to prove that if any such language exists, then Greibach's \emph{hardest context-free language} $\cH$~\cite{greibach1973hardest} also has no PEGs. So the above problem is equivalent to asking for a proof that $\cH$ has no parsing expression grammar. But no PEG is known, even for the much simpler language of \emph{palindromes}. The following questions are both open:

\begin{center}\em
  Can a parsing expression grammar recognise the language of palindromes?

  \medskip
  Is there any linear-time language without a parsing expression grammar?
\end{center}

In fact, the only method we know to prove that a language has no PEG is by using the time-hierarchy theorem of complexity theory \cite{hartmanis1965computational}: using diagonalisation one may construct some language $L_2$ which is decidable, say, in time $n^2$ (by a random-access machine), but not in linear time, and because PEGs can be recognised in linear time using the tabular parsing algorithm of Birman and Ullman \cite{birman1970parsing} (or packrat parsing \cite{ford02:_packr,Ford:A6T0y0WG}), there will be no parsing expression grammar for $L_2$.

This stands in stark contrast with our understanding of, say, context-free languages. In that scenario, one may also construct a language $L_4$ which is decidable in time $n^4$, which cannot be decided in time $n^3$, and hence $L_4$ cannot be context-free (since the CYK algorithm decides any context-free language in time $n^3$, see, e.g., Hopcroft's book \cite{hopcroft2001introduction}). But this brings us no real insight on what it means to be context-free. To understand this, we make use of \emph{pumping lemmas}, and using such lemmas we can easily provide, say, a linear-time-decidable language which is not context-free. A pumping lemma implies a serious limitation on the computational power of context-free languages, which does not apply to universal models of computation, such as Turing machines or random-access machines.

Our current understanding of universal computation, by contrast, is extremely poor. For example, it is a longstanding open problem, to show that linear-time random-access machines cannot be simulated by two-tape Turing machines in linear time, even though it seems intuitive that this should be true. Indeed this problem is well beyond the current state of the art in computational complexity, where such lower-bounds are notoriously difficult to come by. It is also an open problem to provide any context-free language which cannot be decided by a two-tape Turing machine in linear time --- for one-tape Turing machines such a separation is known\footnote{This was first proven for palindromes; see Li and Vitanyi \cite[][\S6.1 and \S6.13]{li1995new}.}.

\medskip
A principal claim of this article is that the recognition procedure underlying parsing expression grammars is, in some sense, ``universal'', and so it will be as difficult to understand as that of a multi-tape Turing machine. A solution to the above questions, thus, may well require a breakthrough in our ability to prove computational complexity lower-bounds.

\medskip
With this in mind, the layout of the article is as follows. In Section \ref{sec:preliminaries}, we provide a formal definition of PEGs, and in Section \ref{sec:example-PEGs} we show a few examples of PEGs with surprising behaviour, and of languages which, unexpectedly, have PEGs. This includes the language of palindromes whose length is a power of two, and it is also shown that PEGs can do a form of counting.

In Section \ref{sec:scaffolding-automata}, we describe a new computational model, the \emph{scaffolding automaton}, and show that it exactly characterises the computational power of PEGs. This is our main result, and provides what we believe to be the right machine model for parsing expression grammars.
We will make good use of this characterisation in Section \ref{sec:applications}, where we show the following results.

\begin{itemize}
\item We revisit the example languages of Section \ref{sec:example-PEGs}, and construct scaffolding automata for them, for the sake of becoming familiar with the model.
\item We show that PEGs are computationally ``universal'', in the following sense:  take any computable function $f:\ZO^\ast\to\ZO^\ast$; then there exists a computable function $g: \ZO^\ast \to \bbN$ such that $$\{ f(x) \$^{g(x)} x \mid x \in \ZO^\ast \}$$ has a PEG. This result may be used to construct a PEG language which is complete for $\PTIME$ under logspace reductions. This stands in contrast to context-free languages, which cannot be $\PTIME$ complete under logspace reductions unless $\PTIME \subseteq \NC_2$.
\item We show that there can be no pumping lemma for PEGs. There is no total computable function $A$ with the following property: for every PEG $G$, there exists $n_0$ such that for every string $x \in \cL(G)$ of size $|x| \ge n_0$, the output $y = A(G, x)$ is in $\cL(G)$ and has $|y| > |x|$.
\item We show that PEGs are strongly non real-time for Turing machines: There exists a language with a PEG, such that neither it nor its reverse can be recognised by any multi-tape online Turing machine which is allowed to do only $o(n/(\log n)^2)$ steps after reading each input symbol.
\end{itemize}

\section{Preliminaries}\label{sec:preliminaries}

In this section we will cover some notation, and give a formal definition of parsing expression grammars.

\paragraph{Notation.} For each $k\in \bbN$, let $(k)_2 \in \ZO^\ast$ be its shortest binary representation, and $(k)_2^r$ to denote the reversal of its shortest binary representation.
An \emph{alphabet} $\Gamma$ is a finite set of symbols such that $\varnothing \notin \Gamma$.
For a natural number $n \ge 0$, we denote $[n] = \{0, \ldots, n\}$, $[n) = \{0, \ldots, n-1\}$, and $(n] = \{1,\ldots,n\}$. We will use $\lambda$ to denote the empty word, and $\eps$ to denote a parsing expression which accepts the empty word.

\begin{definition}\label{def:parsing-expressions}
  Let $\Sigma, \NT$ be two disjoint alphabets; the symbols in $\Sigma$ are called \emph{terminal} symbols, and those in $\NT$ are called \emph{non-terminal} symbols. Then, the set $\cE(\Sigma, \NT)$ of \emph{parsing-expressions over $\Sigma$ and $\NT$} is defined inductively.
  \begin{itemize}
  \item At the base of the induction we have $\Sigma \cup \NT \cup \{ \ACCEPT, \FAIL \} \subseteq \cE(\Sigma, \NT)$.
    \item If $e \in \cE(\Sigma, \NT)$, we will have $\gnot e$ and $\gand e$ in $\cE(\Sigma, \NT)$.
    \item If $e_1, e_2 \in \cE(\Sigma,\NT)$, we will have $e_1 e_2$ and $e_1 / e_2$ in $\cE(\Sigma, \NT)$.
\end{itemize}
\end{definition}

\begin{definition}\label{def:peg}
  A \emph{parsing expression grammar} $\cG$ is a tuple $\langle \Sigma, \mathsf{NT}, R, S\rangle$, where
  \begin{itemize}
  \item $\Sigma$ is an alphabet of so-called \emph{terminal symbols}.
  \item $\mathsf{NT}$ is an alphabet of so-called \emph{non-terminal symbols}, disjoint from $\Sigma$.
  \item $R: \mathsf{NT} \to \cE(\Sigma, \NT)$ is a function defining the \emph{rules} of $\cG$, and associates a $(\Sigma,\NT)$-parsing-expression to each non-terminal symbol.
  \item $S \in \NT$ is the \emph{starting non-terminal}.
  \end{itemize}
\end{definition}

\bsni
When writing down a parsing expression grammar, the notation $A \leftarrow e$ is used to signify $R(A) = e$. The reason one uses the left arrow notation is to emphasise that PEGs correspond to a \emph{recognition procedure}, and are not to be thought of as a generative model.

\bsni
Ford \cite{ford2004parsing} defines parsing expressions that allow for various operations, such as the \emph{zero-or-more repetitions} operator ``{\tt *}'', or the \emph{any character} symbol ``{\tt .}''. As explained in Ford's paper \cite{ford2004parsing}, these operators can be expressed by using the operators appearing in Definition \ref{def:parsing-expressions}, together with the grammars of Definition \ref{def:peg}. This is similar to how one would define such operators using context-free grammars, so we will not explicitly include these operators as part of Definition \ref{def:parsing-expressions}. For the sake of example, the \emph{zero-or-more repetitions} operator $A^{\text{\tt *}}$, applied to a non-terminal $A$, may be replaced by a new non-terminal $\mathsf{Astar}$ together with the rule $\mathsf{Astar} \leftarrow A\; \mathsf{Astar} \;/\; \eps$.

\bsni
\Hand\ \
The \emph{any character} symbol ``{\tt .}'', which we will be using extensively throughout, may be replaced with $(a / b / \ldots)$ for each terminal symbol $a, b, \ldots$ of $\Sigma$. After we define the recognition procedure underlying a parsing expression grammar, in Definition \ref{def:recognition} below, it may be seen that the parsing expression ``$\gnot {\tt .}$'' recognizes exactly the empty string at the end of the input.

\bsni
In order to define a rule $A \leftarrow B / C / \ldots$, we will write rules of the form $A \leftarrow B$, $A \leftarrow C$, \emph{etc}, and say they are \emph{alternatives} of the non-terminal symbol $A$. So, for example, if we say $A \leftarrow B A$ and $A \leftarrow \eps$ are alternatives of $A$, we mean that the rule for $A$ is $R(A) = B A \;/\; \eps$. We will only do this when the order in which the alternatives appear in the rule is indifferent.

\bsni
Each parsing expression grammar defines an associated recognition procedure. This procedure gives an operational meaning to each PEG.

\begin{definition}[Recognition]\label{def:recognition}
  Let $\cG = \langle \Sigma, \NT, R, S\rangle$ be a parsing expression grammar. The \emph{recognition map} is a partial function
  \[
    \rec_\cG: \cE(\Sigma,\NT) \times \Sigma^\ast \to \Sigma^\ast \cup \{ \FAIL \};
  \]
  this map is defined by Algorithm \ref{recognition-procedure} appearing below.
  If $\rec_\cG(e, x) = \FAIL$, we say that expression $e$ \emph{rejects} input $x$; and if $\rec_\cG(e,x) = x'$ outputs a prefix $x'$ of $x$, we say that expression $e$ \emph{accepts} $x$, and \emph{consumes} $x'$. If $\rec_\cG(e, x) = x$, i.e.~$e$ accepts $x$ and consumes all of $x$, then we say the expression $e$ \emph{recognises} $x$. Otherwise $\rec_\cG(e, x)$ is \emph{undefined}, which happens precisely when the recognition procedure entered an infinite loop. We say that $\cG$ is \emph{total} if its recognition map is total, i.e.~if it never enters an infinite loop, on any input.
\end{definition}

\bsni
\Hand\ \  The notions \emph{rejects}, \emph{accepts}, \emph{consumes} and \emph{recognises} will be frequently used throughout the paper, and the reader may refer to the above definition to remember what they mean. It is important to understand that a parsing expression $e$ may accept a string $x$, without consuming all of it. For example the expression $\gand(a a)$ accepts the string $a a$ but consumes no symbol in it.

\bigskip
\setlength{\intextsep}{4pt plus 1.0pt minus 2.0pt}
\begin{algorithm}[h]
    \caption{Recognition Procedure $\rec_\cG(E, x):$}\label{recognition-procedure}
    \begin{algorithmic}[1]
      \small
      \Require{$E \in \cE(\Sigma,\NT), x \in \Sigma^*$}
      \Ensure{$\rec_\cG(E, x) \in \Sigma^* \cup \{ \FAIL \}$}
      
      \If{$E = \eps$}  {\bf return} the empty string $\lambda$
      \ElsIf{$E = \FAIL$}\  {\bf return} $\FAIL$
      \ElsIf{$E = a \in \Sigma$}
      \If{$x = a z$ for some $z$} {\bf return} $a$ {\bf else return} $\FAIL$ \EndIf
      \ElsIf{$E = \gnot e$}
      \If{$\rec_\cG(e, x) = \FAIL$} {\bf return} $\lambda$ {\bf else return} $\FAIL$
      \EndIf
      \ElsIf{$E = \gand e$}
      \If{$\rec_\cG(e, x) \in \Sigma^*$} {\bf return} $\lambda$ {\bf else return} $\FAIL$
      \EndIf
      \ElsIf{$E = e_1 e_2$}
      \If{$\rec_\cG(e_1, x) = y_1 \in \Sigma^\ast$ {\bf and } $x = y_1 z$ {\bf and } $\rec_\cG(e_2, z) = y_2 \in \Sigma^\ast$}
      \State {\bf return} $y_1 y_2$
      \Else\ {\bf return} $\FAIL$
      \EndIf
      \ElsIf{$E = e_1 / e_2$}
      \If{$\rec_\cG(e_1, x) \in \Sigma^\ast$}\ {\bf return} $\rec_\cG(e_1, x)$
      \Else\ {\bf return} $\rec_\cG(e_2, x)$
      \EndIf
      \ElsIf{$E = A \in \NT$}\ {\bf return} $\rec_\cG(R(A), x)$
      \EndIf
    \end{algorithmic}
  \end{algorithm}

\begin{definition}\label{def:PEGclass}
  A total PEG $\cG = \langle \Sigma, \NT, R, S\rangle$ is said to \emph{recognise} the language $\cL(\cG) = \{ x\in \Sigma^\ast \mid \rec_\cG(S, x) = x \}$.

  Then $\PEG$ is the class of languages recognised by total PEGs.
\end{definition}

One consequence of the results in this paper is that no algorithm can decide whether a PEG is total. Ford's original paper \cite{ford2004parsing} defined a notion, that of \emph{well-formed} parsing expression grammar, which was inherited from Birman and Ullman \cite{birman1970parsing}. A well-formed PEG is a PEG which obeys a certain syntactic restriction; this restriction guarantees that the above recognition procedure will not enter an infinite loop (but not all total PEGs are well-formed).

Informally, a PEG is well-formed if it avoids left recursion.
To avoid excessive formalism, in this paper we will not concern ourselves with the formal definition of well-formed PEGs. All the PEGs appearing in this paper are total, and, for the readers familiar with the notion of well-formedness, it will be possible to see that they are also well-formed. Furthermore, every theorem in this paper referring to ``total'' PEGs will still hold if one restricts our attention to ``well-formed'' PEGs.

Furthermore, there is an algorithm which accepts a PEG $\cG$ as input, and outputs a well-formed PEG $\cG'$, such that $\cG'$ recognises the same language as $\cG$ whenever $\cG$ is total. This is akin to the fact that, despite it being undecidable if a given Turing machine runs in time $n^2$, one can take any Turing machine $\cM$ and convert it into a (multitape) Turing machine $\cM'$ which does run in time $n^2$, and which decides the same language as $\cM$ if $\cM$ also runs in time $n^2$ \cite[see][]{balcazar1988structural,arora2009computational}.

\section{Illustrative Examples}\label{sec:example-PEGs}

In this section we will study some examples which were instrumental for us to understand the computational power of the model.

\subsection{Power-Length PEGs}
Our initial expectations for the computational power of PEGs were that we should be able to treat them in a similar way as with context-free grammars, by showing a pumping lemma for them.

This owed not so much to what we knew about the computational power of PEGs --- which already Birman and Ullman \cite{birman1970parsing}, and Ford \cite{ford2004parsing}, had shown surpasses that of CFGs --- but rather to the context in which one studies PEGs: \emph{if PEGs are regarded in the context of formal languages, then we should be able to prove some kind of pumping lemma}. But soon we stumbled on the following example from the PhD thesis of Birman \cite{birman1970tmg}:

\begin{theorem}
  The unary language of words whose length is a power-of-2
  \[
    \cP_2 = \{ a^{2^n} \mid n \ge 0 \}
  \]
  is in $\PEG$.
\end{theorem}

How does this relate to pumping lemmas? The known pumping lemmas are able to produce, given a sufficiently large string $x$ in the language, a strictly larger string $y$, also in the language, which is not much larger --- $|y| \le |x| + O(1)$ is sufficient. But here is a language with a PEG, for which $|y|$ is always at least $2 |x|$. And soon after conjecturing that $c \cdot |x|$ might be sufficient, for some universal constant $c$, one is disabused of that notion by the following generalisation of the above:

\begin{theorem}\label{thm:powers-of-k}
  For every $\ell \in \bbN$, the language
$\cP_\ell = \{ a^{\ell^n} \mid n \ge 0 \}$
is in $\PEG$.
\end{theorem}

\newcommand{\IAmPowerLLength}{\mathsf{IAmPowerLLength}}
\begin{proof}
  Consider the following parsing expression grammar $\cG$:
    \[
      \IAmPowerLLength \leftarrow a \gnot . \quad/\quad \mathsf{Helper} \; \gnot.
    \]
    \[
      \mathsf{Helper} \leftarrow a^{\ell-1}\; \mathsf{Helper}\; a \quad / \quad  a^{\ell-1} (\gand \mathsf{Helper}) a \quad / \quad  a ((\gnot \mathsf{Helper}) a)^{\ell-1}
    \]
    Let us analyse the behaviour of the recognition procedure $\rec_\cG(\mathsf{Helper}, x)$ for each $x \in \{a\}^\ast$. The shortest $x$ to be accepted will be $a^\ell$; this string is accepted via the third alternative of the $\mathsf{Helper}$ non-terminal, and every symbol will be consumed, so $a^\ell$ is recognised by $\mathsf{Helper}$. Then the second string to be accepted will be $a^{\ell-1} a a^{\ell-1}$, via the second alternative --- the first alternative must have failed because it won't find the last $a$. So the second alternative is triggered, but only the first $\ell$-many $a$ symbols will be consumed, leaving $a^{\ell-1}$ symbols unconsumed (hence the string will be ``accepted'', but it won't be ``recognised''). Then the first alternative will trigger for each new sequence of $\ell-1$ $a$s, each time consuming a new $a$ symbol closer to the end of the input. Hence at this point in total we will have consumed $(\ell-1)\ell$ new symbols, which together with the $\ell$ symbols give us $\ell^2$ consumed symbols, and at this point the non-terminal $\mathsf{Helper}$ will have consumed the entire input. Thus $a^{\ell^2}$ is accepted by $\mathsf{Helper}$. Then again the second alternative is triggered, and then the first, until $\ell^3$ symbols are consumed.

    In the end, we conclude that $\mathsf{Helper}$ accepts any string of the form
    \[
      a^{s(\ell-1)} a^s \; z,
    \]
    where the first position of the $a^s$-part is the first position at a power-of-$\ell$ distance from the end of the input, and in this case it consumes the first $s \ell$-many $a$ symbols.
\end{proof}

\subsection{PEG for Sometimes-Palindromes}

One may get a sense for the limitations of parsing expression grammars when trying to produce a PEG for recognising palindromes. One quickly comes to the conjecture that PEGs cannot find the middle bit of the input. In the case of palindromes, we make the following conjecture:

\begin{conjecture}
  The language of even-length palindromes has no PEG, i.e.
  \[
   \mathsf{P} = \{ w w^r \mid w \in \ZO^\ast \} \notin \PEG.
  \]
\end{conjecture}

\noindent
However, the above PEG for $\cP_2$ \emph{is} able to find the middle bit of every string whose length is a power of two. This allows us to prove the following result:
\begin{theorem}\label{thm:sometimes-palindromes}
  The language of palindromes of power-of-two length has a PEG:
  \[
    \mathsf{SP} = \{ w w^r \mid w \in \ZO^{2^n}, n \ge 0 \} \in \PEG.
  \]
\end{theorem}

\newcommand{\IAmPowerTwoLength}{\mathsf{IAmPowerTwoLength}}
\begin{proof}
  The following parsing expression grammar will do:
  \[
    \mathsf{S} \leftarrow \gand( \IAmPowerTwoLength ) \; \mathsf{Palindrome}
  \]
  \[
    \mathsf{Palindrome} \leftarrow \mathsf{P}\gnot. \;/\; 0 0 \gnot . \;/\; 1 1 \gnot .
  \]
  \begin{align*}
    \mathsf{P} \leftarrow \; & \; 0 \; \gnot(\IAmPowerTwoLength)\; \mathsf{P} \;0\\
    / & \; 1 \; \gnot(\IAmPowerTwoLength) \; \mathsf{P} \; 1\\
    / & \; 1 \; \gand(\IAmPowerTwoLength) \;1\\
    / & \; 0 \; \gand(\IAmPowerTwoLength) \; 0
  \end{align*}
  \[
    \IAmPowerTwoLength \leftarrow \mathsf{Helper} \; !.
  \]
  \[
    \mathsf{Helper} \leftarrow \mathsf{Bit}\; \mathsf{Helper}\; \mathsf{Bit} \quad / \quad  \mathsf{Bit} \; \mathsf{Bit}
  \]
  \[
    \mathsf{Bit} \leftarrow 0 / 1
  \]
  As in the proof of Theorem \ref{thm:powers-of-k}, the non-terminal $\IAmPowerTwoLength$ accepts exactly at the positions whose distance from the end-of-input is a positive power of two, and consumes the entire input in that case. Hence the expression $(\gand \IAmPowerTwoLength)$ accepts exactly at positions whose distance from end-of-input is a positive power of two, and when it accepts it will not consume any input. On the other hand the expression $(\gnot \IAmPowerTwoLength)$ accepts exactly at positions which are \emph{not} at positive-power-of-two distance away from the end-of-input.

  The recognition procedure associated with the non-terminal $\mathsf{P}$ now behaves as follows: one of the first two alternatives will be chosen repeatedly, until the first position which is a positive power-of-two is reached; then, at that position, one of the last two alternatives is chosen. (In each case, which of the two alternatives gets chosen is determined by the next bit.) It follows that $\mathsf{P}$ accepts exactly at those positions $i$ such that the input after (and including) position $i$ is of the form:
  \[
    x\, y \, z
  \]
  where $x = y^r$, and the leftmost position after $i$ which is at a positive-power-of-two distance away from the end-of-input, is the first bit of $y$. And when $\mathsf{P}$ accepts such a string $x y z$, $\mathsf{P}$ consumes exactly the prefix $x y$.

  Inspection of the rules for $\mathsf{Palindrome}$ and $\mathsf{S}$ concludes the proof.
\end{proof}

\subsection{PEG for a Counting Language}

The next example will be crucial in Sections \ref{sec:universality} and \ref{sec:peg-vs-online}, for reasons which we will explain in Section \ref{sec:expressive-power}.

\begin{theorem}\label{thm:counting}
  The following \emph{reversed counting language}, over the alphabet $\{0,1,\#,\circ\}$, has a parsing expression grammar:
  \[
    \{ (n)_2^r \circ (n)_{2} \# \; (n-1)_2^r \circ (n-1)_2 \#\; \cdots \;
    \# \; (0)_2^r \circ (0)_2 \#   \mid n \ge 0\}.
  \]
\end{theorem}

\bsni
The characters $\#$ and $\circ$ are part of the input alphabet, and are being used as separators, with no other special meaning. We will call $\#$ the \emph{outer separator}, and $\circ$ the \emph{inner separator}.

\begin{proof}
  The proof relies on the intuition built in the previous two proofs. Roughly speaking, it implements the simple increment-by-one algorithm.

  Let us begin by presenting only part of the grammar. We will omit the rules associated with the non-terminal $\mathsf{AddOneBlock}$, for now. The grammar begins with the rules:
  \[
    \mathsf{Sequence} \leftarrow \gand (\mathsf{AddOneBlock})\; \mathsf{InvertedBlock} \; \mathsf{Sequence} \quad / \quad 0\circ 0\#
  \]
  \[
    \mathsf{InvertedBlock} \leftarrow \mathsf{Inverted} \#
  \]
  \[
    \mathsf{Inverted} \leftarrow 1\; \mathsf{Inverted} \; 1 \quad/\quad 0\; \mathsf{Inverted}\; 0 \quad/\quad \circ
  \]
  The first thing to notice is that $\mathsf{InvertedBlock}$ recognises exactly ``inverted blocks'' of the form $w^r \circ w \#$, where $w \in \ZO^\ast$. Thus the inputs recognised by $\mathsf{Sequence}$ are exactly sequences of inverted blocks which additionally are accepted by the $\mathsf{AddOneBlock}$ non-terminal; the rules for this non-terminal are:
  \[
    \mathsf{AddOneBlock} \leftarrow \mathsf{Bit}^+ \circ \mathsf{AddOneCheck}
  \]
  \[
    \mathsf{AddOneCheck} \leftarrow \mathsf{AddOneDigit}\; \mathsf{AddOneCheck} \quad / \quad \#
  \]
  Now $\mathsf{AddOneBlock}$ accepts strings of the form $x \circ y \#$, such that $x \in \ZO^\ast$, and such that $\mathsf{AddOneDigit}$ accepts the input at every position of $y$. This will be defined in such a way that, at the $i$-th bit of $y$ (starting from the right), $\mathsf{AddOneDigit}$ will accept if and only if the $i$-th bit of $(n+1)_2$ is $y_i$, where $n$ is the number encoded in the following block (i.e.~after the $\#$).

  To enforce this behaviour, we use the following rules:
  \begin{align*}
    \mathsf{AddOneDigit} \leftarrow\; & \; \gand \mathsf{NextIs1} \; \gand \mathsf{Carry} \; 0\\
    /&\; \gand \mathsf{NextIs0} \; \gand \mathsf{Carry} \; 1\\
    /&\; \gand \mathsf{NextIs1} \; \gnot \mathsf{Carry} \; 1\\
    /&\; \gand \mathsf{NextIs0} \; \gnot \mathsf{Carry} \; 0\\
    /&\; \gand \mathsf{NextIsCircle} \; \gand \mathsf{Carry} \; 1
  \end{align*}
  \[
    \mathsf{Carry} \leftarrow . \; \gand \mathsf{NextIs1} \; \gand \mathsf{Carry} \quad /\quad \mathsf{Bit} \; \#
  \]
  The non-terminals $\mathsf{NextIs0}$, $\mathsf{NextIs1}$, and $\mathsf{NextIsCircle}$ will verify that the input symbol in the corresponding position in the next block is a $0$, a $1$ or a $\circ$, respectively. So, for example, if the input after the current position is
  \[
    y_i y_{i+1} \cdots y_k \# x_k \cdots x_{i-1} x_i,
  \]
  then $\mathsf{NextIs0}$ will accept iff $x_i = 0$, $\mathsf{NextIs1}$ will accept iff $x_i = 1$, and $\mathsf{NextIsCircle}$ will accept iff $x_i = \circ$.
   
  It results from this that the non-terminal $\mathsf{Carry}$ accepts if and only if there is a carry at the current position, when we add $1$ to the number after the $\#$ separator: we implement the incremented $1$ by setting the carry to $1$ at the least significant bit, and then the carry propagates as long as the number after the separator has a $1$. Then $\mathsf{AddOneDigit}$ successfully checks a single digit in the increment, in the usual way: a $1$ and a carry sum to $0$, a $0$ and a carry sum to $1$, \emph{etcetera}.

  All we are left to do is defining the auxiliary non-terminals:
  \[
    \mathsf{NextIs0} \leftarrow \mathsf{Bit} \; \mathsf{SameLength} \; 0
  \]
  \[
    \mathsf{NextIs1} \leftarrow \mathsf{Bit} \; \mathsf{SameLength} \; 1
  \]
  \[
    \mathsf{NextIsCircle} \leftarrow \mathsf{Bit} \; \mathsf{SameLength} \; \circ
  \]
  \[
    \mathsf{SameLength} \leftarrow \mathsf{Bit}\; \mathsf{SameLength} \;\mathsf{Bit} \quad /\quad \#
  \]
  \[
    \mathsf{Bit}^+ \leftarrow \mathsf{Bit} \; \mathsf{Bit}^+ \;/\; \mathsf{Bit}
  \]
  \[
    \mathsf{Bit} \leftarrow 0 \;/\; 1 \qedhere
  \]
\end{proof}

Let us here make an important remark. The simple increment-by-one algorithm works by scanning the bits from right to left. However it does not appear to be possible to implement such a right-to-left scanning using PEGs, but left-to-right scanning can be done, and this is what the $\mathsf{NextIs}\ast$ non-terminals are doing, and checking inversion is possible, as shown by the $\mathsf{Inverted}$ non-terminal. So we may implement right-to-left scanning by inverting at each block and then using left-to-right scanning. This trick will be called ``reverse and scan'', and will be used in our simulation of Turing machines by PEGs (in Section \ref{sec:universality}), as well as in our construction of a non-real-time $\PEG$ language (in Section \ref{sec:peg-vs-online}).

\paragraph{Conclusion}
While carefully considering the examples above, one will get a sense that the computational power of PEGs is much greater than it seems at first glance. When considering why and how these examples work, one is slowly drawn to a generalisation of the above: a computational model for languages recognised by parsing expression grammars. This is what we present in the next section.

\section{Scaffolding Automata}\label{sec:scaffolding-automata}

Let us begin by giving an informal description of a scaffolding automaton. Such an automaton is a computing machine which constructs a labelled, directed, acyclic graph of bounded degree, which we call a \emph{scaffold}. At the start of the computation, the graph is a single node with a special end-marker label; this is the \emph{base} of the scaffold. Then as the computation proceeds new input symbols are read and new nodes are added; the node which was last added is called the \emph{top} of the scaffold. At each step of computation, the scaffolding automaton sees a new input symbol, and is allowed to look at a finite-distance neighbourhood of the top; based on the edges which are present, on the labels it sees, on the input symbol it just read, and on the current state of its finite control, the automaton adds a new node to the scaffold (the new top), and chooses the edges of this new node to point to some nodes in the finite-distance neighbourhood it has just observed. This is repeated until all input symbols are read.

\subsection{Formal Definition}

\begin{definition}[Scaffold]\label{def:scaffold}\label{def:path}
  Let $d \ge 1$, $t \ge 0$ be natural numbers, and let $\Gamma$ be an alphabet. An \emph{edge list} of degree $d$ is a tuple
  \[
    e=(e(0),\ldots,e(d-1))\in (\bbN \cup \{ \varnothing \})^d.
  \]
  A \emph{$(d, \Gamma)$-scaffold} of size $t + 1\in \mathbb{N}$ is a labelled multidigraph $S=(V,E,L)$ with set of nodes $V = [t]$, a set of edge lists $E=\{\;e_v\in (\bbN \cup \{\varnothing\})^d\mid v \in [t]\;\})$, where
  \begin{equation*}
    \forall v\in [t]\; \forall i \in [d)\;\quad e_v(i)\in [v] \cup \{\varnothing\},\tag*{(``edges point backwards'')}
  \end{equation*}
  and a labelling function $L: V \to \Gamma \cup \{\varnothing\}$.

  We call $t$ the \emph{top} of the scaffold $S$.  If $e_v(i)=\varnothing$, one says that that \emph{node $v$ is missing edge $i$}, otherwise we say that \emph{edge $i$ is present at node $v$}. If $L(v) = \varnothing$, one says \emph{$v$ is unlabelled}.  Let $\bbS(d, \Gamma)$ be set of all $(d,\Gamma)$-scaffolds (of any length).

  Given a tuple $p \in [d)^k$, and a node $v \in V$ in a $(d,\Gamma)$-scaffold $S = (V,E,L)$, we may inductively define the sequence
  \[
    v_0 = v\text{ and } v_{j+1} =
    \begin{cases}
      e_{v_j}(p_j) & \text{if } v_j \in V,\\
      \varnothing & \text{if } v_j = \varnothing.
    \end{cases}
  \]
  If this sequence has $v_i = \varnothing$ for some $i \in [k]$, we say $p$ is an \emph{invalid path from $v$ in $S$}. Otherwise we say $p$ is a (valid) \emph{path from $v$ to $v_k$ in $S$}.
\end{definition}

\begin{definition}[Neighbourhood]
  Given $S=(V,E, L) \in \bbS(d,\Gamma)$, $k \ge 0$ and $v\in V$, the \emph{$k$-neighbourhood of $v$ in $S$}, $N_k(S, v)$, is given inductively by $N_0(S, v)=L(v)$ and $N_{k+1}(S, v) = (L(v), N_k(S, e_v(0)), \ldots, N_k(S, e_v(d-1)))$, where we set $N_k(S, \varnothing) = \varnothing$.
  
The set of \emph{$k$-neighbourhoods for $(d,\Gamma)$-scaffolds}, $\cN_k(d,\Gamma)$, is the set of partial, $d$-ary, $\Gamma$-labelled trees. It may be inductively defined by letting $\cN_0(d, \Gamma) = \Gamma \cup \{\varnothing\}$ and $\cN_{k+1}(d,\Gamma) = (\Gamma \cup \{\varnothing\})\times(\cN_{k}(d, \Gamma) \cup \{\varnothing \})^d$.

\end{definition}

   \begin{definition}[Scaffolding automaton]
  	 A \emph{scaffolding automaton} $\cA$ is a tuple $\cA = \langle \Sigma, d, \Gamma, k, Q, \delta, q_0, F\rangle$, where,
  \begin{itemize}
  \item $\Sigma$ is an alphabet, called the \emph{input alphabet},
  \item $d \ge 1, k \ge 0$ are natural numbers, called \emph{degree} and \emph{distance}, respectively,
  \item $\Gamma$ is an alphabet, called the \emph{working alphabet},
  \item $Q$ is a finite set of \emph{states},
  \item $q_0 \in Q$ is the \emph{initial state},
    \item $F \subseteq Q$ gives the \emph{accepting states}, and 
    \item the \emph{transition function} is of type
      \[
        \delta:Q\times\Sigma\times \cN_k(d, \Gamma) \to Q \times \Gamma \times ([d)^{\le k}\cup \{\mathsf{SELF}, \varnothing \})^d.
      \]
    \end{itemize}
  \end{definition}

  \noindent
  A scaffolding automaton builds a scaffold while reading the input. The initial scaffold is $S_0 = (\{0\}, \{\}, L)$ where $L(0) = \varnothing$. The transition function $\delta$ transforms a scaffold as follows.

\begin{definition}[Single step of computation] Let $S = ([t], E, L) \in \bbS(d,\Gamma)$, and $\delta$
  be a transition function. For some $q\in Q$ and $\sigma\in \Sigma$, let
  \[
    (q', \gamma, p_0, \ldots, p_{d-1}) = \delta(q,\sigma,N_k(S,t)).
  \]
  The \emph{single-step function} is then given by $\Step_{\delta,\sigma}(q, S) = (q',S')$, where $S' = ([t+1], E', L') \in \bbS(d,\Gamma)$, with $L'(t+1) = \gamma$, $L'(v) = L(v)$ for $v \in [t]$, and $E' = E \cup \{ e_{t+1} \}$, for the edge list $e_{t+1} = (v_0, \ldots, v_{d-1})$, where $v_i$ is obtained by following path $p_i$ from $t$ in $S$ (and equals $\varnothing$ if $p_i$ is an invalid path from $t$ in $S$); if $p_i = \varnothing$, then $e_{t+1}(i) = \varnothing$ also, and if $p_i = \mathsf{SELF}$, then $e_{t+1}(i) = t+1$.
\end{definition}

\noindent
We now formally define how the computation proceeds.
\begin{definition}
  Let $\cA = \langle \Sigma, d, \Gamma, k, Q, \delta, q_0, F\rangle$ be a scaffolding automaton, and $x = \sigma_1 \cdots \sigma_n \in \Sigma^n$. Then the \emph{computation of $\cA$ on $x$}, denoted $\cA(x)$, is a sequence
  \[
    \cA(x) = ((q_0,S_0),(q_1,S_1), \ldots, (q_n, S_n)) \in (Q\times \bbS(d, \Gamma))^{1+n}.
  \]
  Having defined $(q_i,S_i)$ up to some $i < n$ --- notice that $q_0$ is the initial state and $S_0$ is the initial scaffold --- we let $(q_{i+1},S_{i+1}) = \Step_{\delta, \sigma_{i+1}}(q_i,S_i)$.
\end{definition}

\begin{definition}  Let $\cA = \langle \Sigma, d, \Gamma, k, Q, \delta, q_0, F\rangle$ be a scaffolding automaton, and $x = \sigma_1 \cdots \sigma_n \in \Sigma^n$. Let $\cA(x) = ((q_0,S_0),(q_1,S_1), \ldots, (q_n, S_n))$ be the computation of $\cA$ on $x$.
  We say that $\cA(x)$ is \emph{accepting} if $q_{n} \in F$; otherwise we say it is \emph{rejecting}.
  This defines the \emph{language decided by $\cA$}:
  \[
    \cL(\cA) = \{ x \in \Sigma^\ast \mid \cA(x) \text{ is accepting} \}.
  \]
\end{definition}

\subsection{Illustrative Examples, Revisited}\label{sec:expressive-power}

We will soon prove that a language has a parsing expression grammar if and only if its reverse is decided by a scaffolding automaton ---  this is Theorem \ref{thm:peg-automata} of Section \ref{sec:equivalence-PEGs}. However, in order to become more familiar with the model, let us begin by directly constructing scaffolding automata for the reverse of the languages seen in Section \ref{sec:example-PEGs}.

For each $\ell \in \bbN$, the power-length language $\cP_\ell^r = \cP_\ell = \{ a^{\ell^n} \mid n \ge 0 \}$ is its own reversal, so let us construct a scaffolding automaton $\cA_\ell$ which decides $\cP_\ell$. Informally, an automaton for $\cP_\ell$ behaves as follows. The automaton makes sure that every node in the scaffold has an edge to the previous node. It first accepts after reading the first $a$, and then after reading the first $\ell$-many $a$'s --- so it accepts $a$ and $a^\ell$. From this point onward a second edge will be maintained that goes backward in the scaffold; we call this edge the \emph{backtracking edge}; the idea is that for each $\ell - 1$ new symbols read, the backtracking edge in the new top node will be moved a single position backwards (towards the base of the scaffold); once the backtracking edge reaches the base, the automaton enters an accepting state and again points the backtracking edge to the new top. This way, the next accepted string will have $\ell$-times as many symbols as the previous accepted string.\footnote{Because $\ell^{k} = \ell^{k-1} + \ell^{k-1}(\ell - 1)$.
}

Let us translate this informal description to the formal definitions given in the previous section. This will be the only scaffolding automaton for which we will do such a translation.

The scaffolding automaton for $\cP_\ell$ is given by $\cA_\ell = \langle \Sigma = \{a \}, d = 2, \Gamma = \{ \boxtimes, \Box \}, k = 2, Q, \delta, q_0, F = \{q_1, q_\ell, q''_{\ell-1}\}\rangle$, where $Q = \{q_0, q_1, \ldots, q_\ell,$ $q'_1, \ldots, q'_{\ell-1},$ $q''_1, \ldots, q''_{\ell - 1}\}$. The degree $d$ equals $2$, and at each node in the scaffold edge $0$ will always point to the previous node, and edge $1$ will be the backtracking edge. We will use wildcards when describing elements of $\cN_k(\Gamma, d)$, so for example $\ast$ means \emph{any element of $\cN_k(\Gamma, d)$} and
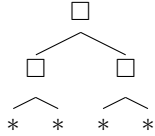
\[
  \small\Tree [.$\Box$ [.$\Box$ $\ast$ $\ast$ ] [.$\Box$  $\ast$ $\ast$ ] ]
\]
means any element of $\cN_k(\Gamma, d)$ (which consists of trees of depth $2$, not trees of depth $1$) whose topmost three nodes are labelled as in the picture above.

\newcommand{\pagedifference}[2]{\number\numexpr\getpagerefnumber{#2}-\getpagerefnumber{#1}\relax}
\newcommand{\ifsamepage}[4]{\ifnum \pagedifference{#1}{#2}=0 #3\else #4\fi}

\bsni
The transition function for $\cA_\ell$ may now be defined. In \ifsamepage{fig:A2-a10}{fig:A3-a10}{page \pageref{fig:A2-a10}}{pages \pageref{fig:A2-a10} and \pageref{fig:A3-a10}} below, we include the diagrams of the two scaffolds resulting from executing $\cA_2$ and $\cA_3$ on the string $a^{10}$. It might be helpful to follow those pictures, to get a sense of how $\cA_\ell$ works.

\begin{itemize}
\item If we are in the initial state and scaffold, the new top will point to the base, will be labelled by $\boxtimes$, and we move to state $q_1$: \[
  \delta\left(q_0, a, \ast\right) = (q_1, \boxtimes, \lambda, \varnothing).
\]
Above, $\lambda$ denotes the empty path, i.e., it is the path to the top node. This edge, edge number $0$, will always be set in this way, so that we may always refer to the previous top node by following edge $0$. The label $\boxtimes$ will be used to distinguish the first node from the rest.
\item We then count $\ell-1$ symbols, as follows: For every $i \in \{ 1, \ldots, \ell-1 \}$ we set \[
    \delta\left(q_i, a, \ast \right) = (q_{i+1}, \Box, \lambda, \varnothing).
\]
\item The state $q_\ell$ is accepting. The next symbol --- symbol number $\ell + 1$ --- triggers the beginning of two nested loops, the \emph{outer loop} and the \emph{inner loop}.
  As we begin the inner loop we point the backtracking edge to the current node in the scaffold (given by the empty path $\lambda$):
  \[
    \delta\left(q_\ell, a, \ast \right) = (q'_{1}, \Box, \lambda, \lambda).
  \]
The inner loop will loop between the states $q'_1, \ldots, q'_{\ell - 1}$, in such a way that, for each sequence of $\ell-1$ input symbols, the backtracking edge is moved backwards a single position in the scaffold. This happens until the backtracking edge reaches the node immediately before the base of the scaffold, at which point we enter the state $q''_1$, which runs the inner loop one last time until reaching state $q''_{\ell-1}$, which is accepting; at state $q''_{\ell-1}$, we ``reset'' the backtracking edge, and we restart the inner loop at $q'_1$. The outer loop consists of this resetting and restarting of the inner loop.

  Let us implement the inner and outer loops. The inner loop counts $\ell-1$ symbols, as follows: for every $i \in \{ 1, \ldots, \ell-2 \}$ we set \[
    \delta\left(q'_i, a, \ast \right) = (q'_{i+1}, \Box, \lambda, (1)).
\]
 When we have finished the inner cycle but have still not found the $\boxtimes$-marked node, we move the backtracking edge backwards, and loop the inner cycle:
  \[
    \delta\left(q'_{\ell-1}, a,
      \begin{array}{c}
\small\Tree [.$\Box$ [.$\Box$ $\ast$ $\ast$ ] [.$\Box$  $\Box$ $\ast$ ] ]
      \end{array} \right) = (q'_1, \Box, \lambda, (1, 0)).
  \]
\item Eventually the top node sees node $1$ of the scaffold at distance $2$ through the backtracking edge --- which we may detect since node $1$ is labelled with $\boxtimes$ instead of $\Box$. At this point we will finish running the inner loop using the $q'$ states, and then run it one last time using the $q''$ states, which behave just like the $q'$ states, except that $q''_{\ell - 1}$ is an accepting state whereas $q'_{\ell - 1}$ is not, and $q''_{\ell-1}$ resets the backtracking edge.

  This is implemented by setting
  \[
    \delta\left(q'_{\ell - 1}, a,
      \begin{array}{c}
        \small\Tree [.$\Box$ [.$\Box$ $\ast$ $\ast$ ] [.$\Box$  $\boxtimes$ $\ast$ ] ]
      \end{array}
    \right) = (q''_1, \Box, \lambda, (1, 0)),
  \]
  and, for each $i \in \{1, \ldots, \ell - 2\}$,
  \[
    \delta\left(q''_i, a, \ast \right) = (q''_{i+1}, \Box, \lambda, (1)),
  \]
  and finally
  \[
    \delta\left(q''_{\ell-1}, a, \ast \right) = (q'_1, \Box, \lambda, \lambda).
  \]
Compare $q''_{\ell-1}$ with $q'_{\ell - 1}$: $q''_{\ell-1}$ is an accepting state whereas $q'_{\ell - 1}$ is not, and $q''_{\ell-1}$ resets the backtracking edge, whereas $q'_{\ell-1}$ moves the backtracking edge one node backwards.
\end{itemize}

In the setup above, each run of the outer cycle consumes $\ell-1$-times as many symbols as the previous run, thus multiplying the total number of consumed symbols by $\ell$. For example, let us picture the run of $\cA_2$ on the string $a^{10}$.

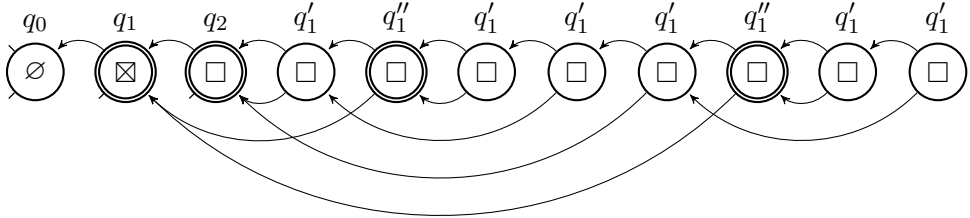
\begin{center}
\begin{tikzpicture}[node distance=1.2cm,>=stealth',bend angle=45,auto]\label{fig:A2-a10}

  \tikzstyle{place}=[circle,thick,draw=blue!75,fill=blue!20,minimum size=6mm]
  \tikzstyle{red place}=[place,draw=red!75,fill=red!20]
  \tikzstyle{transition}=[rectangle,thick,draw=black!75,
  			  fill=black!20,minimum size=4mm]

  \tikzstyle{snode}=[circle,thick,draw=black,minimum size=6mm]
  
  \begin{scope}
\node [snode,label={$q_0$}] (n0) {$\varnothing$};

    \draw[line width=.5pt] (n0) -- +(135:0.5cm);
    \draw[line width=.5pt] (n0) -- +(225:0.5cm);

    \node [snode,label={$q_1$},accepting] (n1) [right of=n0] {$\boxtimes$}
    edge [post,bend right] (n0);
    \draw[line width=.5pt] (n1) -- +(225:0.5cm);
        
    \node [snode,label={$q_2$},accepting] (n2) [right of=n1] {$\Box$}
    edge [post,bend right] (n1);
    \draw[line width=.5pt] (n2) -- +(225:0.5cm);

    \node [snode,label={$q'_1$}] (n3) [right of=n2] {$\Box$}
    edge [post,bend right] (n2)
    edge [post,bend left] (n2);
    
    \node [snode,label={$q''_1$},accepting] (n4) [right of=n3] {$\Box$}
    edge [post,bend right] (n3)
    edge [post,bend left] (n1);
    
    \node [snode,label={$q'_1$}] (n5) [right of=n4] {$\Box$}
    edge [post,bend right] (n4)
    edge [post,bend left] (n4);
    
    \node [snode,label={$q'_1$}] (n6) [right of=n5] {$\Box$}
    edge [post,bend right] (n5)
    edge [post,bend left] (n3);
    
    \node [snode,label={$q'_1$}] (n7) [right of=n6] {$\Box$}
    edge [post,bend right] (n6)
    edge [post,bend left] (n2);
    
    \node [snode,label={$q''_1$},accepting] (n8) [right of=n7] {$\Box$}
    edge [post,bend right] (n7)
    edge [post,bend left] (n1);
    
    \node [snode,label={$q'_1$}] (n9) [right of=n8] {$\Box$}
    edge [post,bend right] (n8)
    edge [post,bend left] (n8);
    
    \node [snode,label={$q'_1$}] (n10) [right of=n9] {$\Box$}
    edge [post,bend right] (n9)
    edge [post,bend left] (n7);
  \end{scope}
\end{tikzpicture}
\end{center}

In the picture, the upper edge points to the previous node, and the lower edge is the backtracking edge. The state of the automaton when reading each node of the scaffold appears above the node, and the node is drawn as a double circle if this state is an accepting state. As required, the automaton accepts after seeing $1$, $2$, $4$, and $8$ symbols. 

For further illustration, let us picture the run of $\cA_3$ on $a^{10}$:

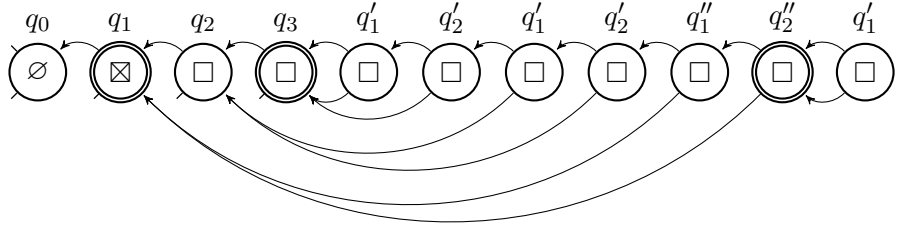
\begin{center}
\begin{tikzpicture}[node distance=1.1cm,>=stealth',bend angle=45,auto]\label{fig:A3-a10}

  \tikzstyle{place}=[circle,thick,draw=blue!75,fill=blue!20,minimum size=6mm]
  \tikzstyle{red place}=[place,draw=red!75,fill=red!20]
  \tikzstyle{transition}=[rectangle,thick,draw=black!75,
  			  fill=black!20,minimum size=4mm]

  \tikzstyle{snode}=[circle,thick,draw=black,minimum size=6mm]
  
  \begin{scope}
\node [snode,label={$q_0$}] (n0) {$\varnothing$};

    \draw[line width=.5pt] (n0) -- +(135:0.5cm);
    \draw[line width=.5pt] (n0) -- +(225:0.5cm);

    \node [snode,label={$q_1$},accepting] (n1) [right of=n0] {$\boxtimes$}
    edge [post,bend right] (n0);
    \draw[line width=.5pt] (n1) -- +(225:0.5cm);
        
    \node [snode,label={$q_2$}] (n2) [right of=n1] {$\Box$}
    edge [post,bend right] (n1);
    \draw[line width=.5pt] (n2) -- +(225:0.5cm);

    \node [snode,label={$q_3$},accepting] (n3) [right of=n2] {$\Box$}
    edge [post,bend right] (n2);
    \draw[line width=.5pt] (n3) -- +(225:0.5cm);
    
    \node [snode,label={$q'_1$}] (n4) [right of=n3] {$\Box$}
    edge [post,bend right] (n3)
    edge [post,bend left] (n3);
    
    \node [snode,label={$q'_2$}] (n5) [right of=n4] {$\Box$}
    edge [post,bend right] (n4)
    edge [post,bend left] (n3);
    
    \node [snode,label={$q'_1$}] (n6) [right of=n5] {$\Box$}
    edge [post,bend right] (n5)
    edge [post,bend left] (n2);
    
    \node [snode,label={$q'_2$}] (n7) [right of=n6] {$\Box$}
    edge [post,bend right] (n6)
    edge [post,bend left] (n2);
    
    \node [snode,label={$q''_1$}] (n8) [right of=n7] {$\Box$}
    edge [post,bend right] (n7)
    edge [post,bend left] (n1);
    
    \node [snode,label={$q''_2$},accepting] (n9) [right of=n8] {$\Box$}
    edge [post,bend right] (n8)
    edge [post,bend left] (n1);
    
    \node [snode,label={$q'_1$}] (n10) [right of=n9] {$\Box$}
    edge [post,bend right] (n9)
    edge [post,bend left] (n9);

\end{scope}
\end{tikzpicture}
\end{center}

\bsni
We started by describing the behaviour for $\cA_\ell$ in some detail, and then provided a fully formal specification. We will now limit ourselves to describing the behaviour in \emph{sufficient} detail, so that the reader may be convinced that a fully formal specification may also be done.

\bsni
Let us now sketch the scaffolding automata for the remaining two examples of Section \ref{sec:example-PEGs}.

Recognising the language of palindromes of power-two length (which also is its own reversal) uses the same idea of maintaining a backtracking edge, and it is similar to the $\ell = 2$ case of the implementation just shown. The backtracking edge is used not only to ensure that the length of the input is a power of two, but is also used to compare the last read symbol with its corresponding symbol. The corresponding symbol, as it turns out, is exactly the symbol under the backtracking edge, as may be verified by the reader by inspecting the run of $\cA_2$ on $a^{10}$, pictured above. In order to make this comparison, thus, the scaffolding automaton may simply label each node with the symbol which was read at that position, and then compare the label of the node under the backtracking edge with the symbol which is now being read. The automaton remembers any violation of this requirement in its finite control, and at each power-of-two length, it accepts if and only if no violation was found.

A scaffolding automaton for recognising the counting language works as follows. The first item in the sequence is of fixed finite length and thus may be recognised --- $\#0^r \circ 0$. Then noticing that if we have recognised the sequence up to $\cdots (n-1)_2^r\circ (n-1)_2\#$ and have an edge pointing to the rightmost bit of $(n-1)_2$, then we may verify, one by one from left-to-right, the bits of $(n)_2^r$ by the usual algorithm for addition. Then we must see a $\circ$, and, having kept an edge pointing to the rightmost bit of $(n)_2^r$, we may now recognise a reversal of $(n)_2^r$, i.e.~$(n)_2$. Then we must see a $\#$. So we have now recognised $\cdots (n)_2^r \circ (n)_2\#$, and we repeat.

This trick, which we have called \emph{reverse and scan}, will be used in the proofs of Theorems \ref{thm:universality} and \ref{thm:non-real-time}.

\subsection{Equivalence with PEGs}\label{sec:equivalence-PEGs}
The rest of this section is devoted to proving that scaffolding automata exactly characterise parsing expression grammars: \begin{theorem}\label{thm:peg-automata}
  A language $L \subseteq \Sigma^\ast$ is in $\PEG$ if and only if its reverse $L^r$ is decided by some scaffolding automaton.
\end{theorem}

The question of whether PEG languages are closed under reverse now arises quite naturally. We conjecture that they are not, but Theorem \ref{thm:universality} below suggests it will be very hard to prove such a result.

\begin{proof}[Proof of Theorem \ref{thm:peg-automata}, necessary direction]
  We begin by proving that a parsing expression grammar for a language $L \subseteq \Sigma^\ast$ gives rise to a scaffolding automaton for $L^r$. A reader who is familiar with the tabular parsing algorithm of Birman and Ullman \cite{birman1970parsing} for TDPLs should be able to easily see that a scaffolding automata can simulate this algorithm (the edges will correspond to entries in the table). Since Ford \cite{ford2004parsing} has shown TDPLs are equivalent to PEGs, that suffices for obtaining the result.

  But Ford's proof of equivalence between PEGs and TDPLs is complex and delicate, whereas scaffolding automata are powerful enough to simulate PEGs directly. So we will prove the result here in full.

   Let $\cG = \langle \Sigma, \NT, R, S\rangle$ be a total parsing expression grammar. Without loss of generality, we may assume that every rule of $\cG$, has one of the forms:
 \begin{itemize}
 \item $A \leftarrow \ACCEPT$, $A \leftarrow \FAIL$, or $A \leftarrow t$, with $A \in \NT$ a non-terminal symbol and $t \in \Sigma$ a terminal symbol.
 \item $A \leftarrow \gnot B$, $A \leftarrow \gand B$ with $A,B \in \NT$.
 \item $A \leftarrow B C$, $A \leftarrow B / C$ with $A,B,C \in \NT$. 
 \end{itemize}
Indeed, any grammar may be converted into the form above by replacing sub-expressions with new non-terminal symbols.\footnote{For example, one would convert the rule $A \leftarrow \gand B C D / E F / \gnot G$ to the rules $A \leftarrow A_1 / A_3$, $A_1 \leftarrow B_1 A_2$, $B_1 \leftarrow \gand B$, $A_2 \leftarrow C D$, $A_3 \leftarrow A_4 / A_5$, $A_4 \leftarrow E F$ and $A_5 \leftarrow \gnot G$.}

\medskip\noindent
We then construct a scaffold automaton $\cA = \langle \Sigma, d, \Gamma, k, Q, \delta, q_0, F\rangle$, where
  \begin{itemize}
  \item $d = |\NT|$ and $k = |\NT|$.
  \item $\Gamma = \{ \Box \}$, as we will use a single label, to distinguish the end of the input from the remaining nodes.
  \item $Q = \{q_{\text{yes}}, q_{\text{no}} \}$, as we will use only two states, which will behave identically except that only one is accepting.
  \item $q_0 = q_{\text{yes}}$ if $\lambda \in \cL(G)$ and $q_0 = q_{\text{no}}$ otherwise.
  \item $F = \{q_{\text{yes}}\}$.
  \end{itemize}
 For $q\in Q$, $\sigma\in\Sigma$ and $N = (V, E, L) \in \cN_k(d,\Gamma)$, the transition function has
  \[
   \delta(q, \sigma, N) = (q', \Box, p_0, \ldots, p_{d-1}),
  \]
  defined as follows. Fix some ordering of $\NT$, and if $A$ is the $i$-th non-terminal symbol in $\NT$, let us use $p_A$ in place of $p_i$. Then:
  \begin{itemize}
  \item If $A \leftarrow \ACCEPT$, set $p_A = \mathsf{SELF}$, i.e., create a self loop in the new top node.
  \item If $A \leftarrow \FAIL$, or $A \leftarrow \sigma'$ with $\sigma' \neq \sigma$, then set $p_A = \varnothing$ --- the new top node will be missing edge $A$.
  \item If $A \leftarrow \sigma$, then set $p_A = \lambda$, i.e., create an edge from the new top to the previous top node.
  \item If $A \leftarrow \gnot B$, then we must first compute $p_B$, and then we set $p_A = \mathsf{SELF}$ if $p_B = \varnothing$, and $p_A = \varnothing$ otherwise.
  \item If $A \leftarrow \gand B$, then we must first compute $p_B$, and then we set $p_A = \mathsf{SELF}$ if $p_B \neq \varnothing$, and $p_A = \varnothing$ otherwise.
  \item If $A \leftarrow B C$, then we must first compute $p_B$; if $p_B = \varnothing$, then we set $p_A = \varnothing$ also; otherwise $p_B$ is a path to some node $v_B$ in $N$; this node will have some edge to $v_{BC} = e_{v_B}(C)$ in $N$ corresponding to $C$; we then let $p_A$ be a path to $v_{BC}$, which is one edge longer than $p_B$. This is where we require $k \ge |\NT|$.\footnote{It may be proven by induction on $|\NT|$ that whenever we set an edge of the new top node, it will be at a distance no greater than $|\NT|$ from the previous top node of the scaffold. Indeed, the only rule which may cause the required distance to increase is the concatenation rule $A \leftarrow B C$. In this case, when the edge $p_B$ points to a node $v_B$ which is a distance $i$ from the previous top node in the scaffold, then $p_A$ will point to the same node $v_{BC}$ as the edge $e_{v_B}(C)$ of $v_B$ corresponding to the non-terminal $C$. So the distance from the previous top node to $v_{BC}$ is now the distance to $v_B$ plus one, i.e., $i+1$. Since, as we argue later, there are no circular dependencies, the maximum distance is then $|\NT|$.}
  \item If $A \leftarrow B / C$, then we must first compute $p_B$ and $p_C$, and then we set $p_A = p_B$, if $p_B \neq \varnothing$, and otherwise we set $p_A = p_C$.
  \end{itemize}

  In the above procedure, we may assume that $p_B$ and $p_C$ are computed before $p_A$, when the rule for $A$ depends on $B$ and $C$. This is because the dependencies of the above procedure (when we say ``we must first compute \ldots'') correspond exactly to the subroutine calls of the recognition procedure $\rec_\cG$. Hence, if we have a cyclic dependency above this will cause $\rec_\cG$ to enter an infinite loop, and our assumption that $\cG$ is total implies that this never happens on any input. Hence if at some point a cyclic dependency is triggered, e.g. ``before computing $p_A$ we must first compute $p_B$ and before computing $p_B$ we must compute $p_A$'', then it may safely be ignored by setting the edge $p_A = \varnothing$, since we are guaranteed, by the totality of $\cG$, that $\rec_\cG$ will not be called for the non-terminal $A$ at this position, on any input.\footnote{Incidentally, it is based on this observation that one may convert a total PEG $\cG$ into an equivalent well-formed PEG. See the discussion after Definition \ref{def:PEGclass}.}

  \bsni
  The above definition ensures that the following property always holds:
  \begin{claim}\label{claim:pegscaffold}
    Let $x^r = x_n \cdots x_1 \in \Sigma^n$ and consider the scaffold $S = (V, E, L)$ obtained at the last step of the computation of $\cA$ on $x^r$. Then the edge of the top node $n \in V$ corresponding to the non-terminal $A \in \NT$ will be present if and only if the corresponding parsing expression $R(A)$ accepts $x = x_1 \cdots x_n$. When present, this edge will point to the position of $x^r$ corresponding to the symbol after $\rec_\cG(R(A), x)$. I.e., if $|\rec_\cG(R(A), x)| = \ell \ge 0$ is the number of consumed symbols, then $e_n \in E$ has $e_n(A) = n-\ell$.
  \end{claim}

  Having defined how we create the new top node, it suffices to explain how the new state $q'$ is chosen. We will set $q' = q_{\text{yes}}$ if the new edge $e_t(S)$, where $t$ is the new top node, and $e_t(S)$ is the edge corresponding to the starting non-terminal of $\cG$, has been set to equal a node with empty label, i.e.~if $L(e_n(S)) = \varnothing$.  We set $q' = q_{\text{no}}$ otherwise. Since only the base of the scaffold has an empty label, we will be in an accepting state if and only if $S$ consumes the entire input seen thus far. By Claim \ref{claim:pegscaffold} it follows that $\cL(\cA) = \cL(\cG)$.
\end{proof}

\begin{proof}[Proof of Theorem \ref{thm:peg-automata}, sufficient direction]
  Now let $\cA = \langle \Sigma, d, \Gamma, k, Q, \delta, q_0, F\rangle$ be a scaffolding automaton accepting the language $L$. Assume without loss of generality (by duplicating states) that $\cA$ is only in the initial state $q_0$ at the very beginning of the computation, and never re-enters it after reading the first symbol.

  We construct a parsing expression grammar $\cG = \langle \mathsf{NT}, \Sigma, R,S\rangle$ recognising $L^r$ . The grammar $\cG$ will have the following non-terminals:
  \begin{itemize}
  \item For each $q \in Q$, we have a non-terminal $\HasState(q)$.
  \item For each $\gamma \in \Gamma$, we have a non-terminal $\HasLabel(\gamma)$.
  \item For each $N \in \cN_k(d,\Gamma)$, we have a non-terminal $\HasNeighbourhood(N)$.
  \item For each $p \in [d)^{\le k}$, we have a non-terminal $\Path(p)$.
  \item The initial non-terminal of the grammar is $\AutomatonAccepts$.
  \end{itemize}

  \begin{tcolorbox}[breakable]
    Now we will define various grammar rules, of the form $$N \leftarrow N_1 \;/\; N_2 \;/\; N_3 \;/\; \ldots,$$ where $N$ is one of the non-terminals $\HasState(q)$, $\HasLabel(\gamma)$, \emph{etcetera}, and $N_1, N_2, \ldots$ are parsing expressions.

    Below, when we say that we ``add an alternative $N \leftarrow E$'', we mean that the rule corresponding to the non-terminal $N$ should have the parsing expression $E$ appearing as one of the parsing expressions $N_i$ on the right-hand side. If no alternative was added in this process, for a given non-terminal $N$, then the rule corresponding to $N$ is instead $N \leftarrow \FAIL$.
    
    So, for example, if during the proof we add the alternative $N \leftarrow A$, the alternative $M \leftarrow B$, then the alternative $N \leftarrow C$, and no other alternatives were added, then the resulting grammar will have the rules $N \leftarrow A \;/\; C$ and $M \leftarrow B$, and for every non-terminal $O$ other than $N$ and $M$, we will have the rule $O \leftarrow \FAIL$.

    This allows us to specify how each transition of the scaffolding automaton affects the different rules appearing in the grammar. If we had to specify each rule of the grammar completely, then we would need to define the rules of the grammar in a fixed order with respect to the non-terminal appearing on the left side, which would obscure the idea behind the construction.
  \end{tcolorbox}
  
  Let $\Sigma = \{\sigma_1, \sigma_2, \ldots \}$ give the (finitely-many) symbols of $\Sigma$.  The rules of the grammar are defined as follows. We have the rule
  \[
    \HasState(q_0) \leftarrow \; \text{\tt !} \; ( \sigma_1 \;/\; \sigma_2 \;/\; \ldots )
  \]
  and if $N_0$ is the trivial neighbourhood containing a single unlabelled node with no edges (i.e.~the neighbourhood of the top node of the initial scaffold), we also have the rule
  \[
    \HasNeighbourhood(N_0) \leftarrow \; \text{\tt !} \; ( \sigma_1 \;/\; \sigma_2 \;/\; \ldots )
  \]
  This ensures that the end of the input of the grammar (which is the beginning of the input of the automaton) matches the initial state and neighbourhood.

  \medskip\noindent Now for each possible $q \in Q$, $\sigma \in \Sigma$, and $N \in \cN_k(d, \Gamma)$, we have a transition
  \[
   \delta(q, \sigma, N) = (q', \gamma, p_0, \ldots, p_{d-1}).
 \]
 Recall that this transition means ``if the scaffolding automaton is in state $q$, reads input symbol $\sigma$, and the neighborhood of the current top node is $N$, then it will move to state $q'$, and create a new top node with label $\gamma$, with edges given by the paths $p_0, \ldots, p_{d-1} \in [d)^{\le k}\cup \{\mathsf{SELF}, \varnothing \}$.''

 \medskip\noindent
  Let us write $\Transition(q, \sigma, N)$ as an abbreviation for the parsing expression $$\text{\tt \&}( \sigma \; \HasState(q) ) \; \text{\tt \&}( \sigma \; \HasNeighbourhood(N) ).$$ 
  We then add the alternative
  \[
    \HasState(q') \leftarrow \Transition(q, \sigma, N).
  \]
  These alternatives will be added for every transition given by $\delta$.  It will follow, by induction on the length of the input string, that $\HasState(q)$ will accept the string $x_{i}\cdots x_1$ if and only if the computation $\cA(x_1 \cdots x_i)$ ends in state $q$; even when it accepts, $\HasState(q)$ will never consume any input. Let $F = \{f_1, f_2, \ldots \}$ give the (finitely-many) accepting states. We then naturally have the rule
  \[
    \AutomatonAccepts \leftarrow (\HasState(f_1) \;/\; \HasState(f_2) \;/\; \ldots) \; \text{\tt .*}
  \]
  \medskip\noindent Then let $\lambda \in [d)^0$ be the sequence of length $0$.  We add the alternative $ \Path(\lambda) \leftarrow \ACCEPT, $ i.e., $\Path(\lambda)$ is always accepted and consumes no input.  Now take a sequence $i p \in [d)^{1 + \ell}$ of length $1 + \ell \ge 1$; then if $p_i \notin \{ \varnothing, \mathsf{SELF} \}$, we add the alternative
  \[
    \Path(i p) \leftarrow \Transition(q, \sigma, N) \;\; \sigma \;\; \Path(p_i) \;\; \Path(p)
  \]
  If $p_i = \varnothing$, we instead add the alternative:
  \[
    \Path(i p) \leftarrow \Transition(q, \sigma, N) \;\; \FAIL
  \]
  And if $p_i = \mathsf{SELF}$, we instead add the alternative:
  \[
    \Path(i p) \leftarrow \Transition(q, \sigma, N) \;\; \Path(p)
  \]
  It will follow by induction that the non-terminal $\Path(p)$ will accept the string $x_i \cdots x_1$ if and only if path $p$ goes from the top of the scaffold in the computation $\cA(x_1 \cdots x_{i})$, i.e.~from node $i$ in that scaffold, to some node $j \le i$. And, if the non-terminal $\Path(p)$ accepts $x_i \cdots x_1$, it will consume the input exactly up to (but not including) position $j$, i.e., it will consume the string $x_{i} \cdots x_{j+1}$ (the entire string will be consumed if $j = 0$, i.e., if the edge points to the base of the scaffold). Finally, we add the alternative
  \[
    \HasLabel(\gamma) \leftarrow \Transition(q, \sigma, N)
  \]

  \medskip\noindent The above alternatives may be added in any order, since the various conditions $\Transition(q,\sigma,N)$ are disjoint.
  The following observation is \emph{crucial} to understand why the above definitions are well-founded: the expression $\Transition(q,\sigma,N)$ uses $\HasState$ and $\HasNeighbourhood$ non-terminals, but \emph{only after consuming symbol $\sigma$}; so the accepting/consuming of the various non-terminals depends on the accepting/consuming of the same non-terminals, but in prior positions of the input, where this has already been determined.

  All we are left to do is explain how each $\HasNeighbourhood$ is defined. But notice that knowing whether the top of a scaffold has a certain neighbourhood consists of checking that certain paths exist, and that the nodes under these paths have certain labels, and that certain other paths do not exist. For example, if we wish to check for the neighbourhood $N \in \cN_2(2, \{\Box, \boxtimes\})$ where the top node is labelled $\Box$, the second edge of the top node leads to a child labelled $\Box$ and that child has itself a child labelled $\boxtimes$ on its first edge, i.e., if $N$ is the neighbourhood:
  \[
    \small\Tree [.$\Box$ [ ] [.$\Box$ $\boxtimes$ [] ] ]
  \]
  we then have the rule:
  \begin{align*}
    \mathsf{Nei}\mathsf{ghbourhood}(N) & \leftarrow \\
                & \gand \HasLabel(\Box) \\
                & \gnot \Path(0) \;\; \text{\tt \&}\Path(1)\\
                & \gand (\Path(1) \;\;\HasLabel(\Box))\\
                & \gand \Path(1,0)\;\; \gnot \Path(1, 1)\\
                & \text{\tt \&}(\Path(1,0 ) \;\; \HasLabel(\boxtimes))
  \end{align*}
  With this observation the proof is now complete.
\end{proof}

\bsni
We would like to make the following remark. It may be observed in the grammar above, which simulates a given scaffolding automaton, that the different alternatives may all be added in any order, since they cover disjoint cases. The reader should now suspect that the prioritized choice operator $/$ may, after all, be replaced by the usual disjunction operator $|$ from context-free grammars. This is entirely correct, since $A \;/\; B$ is equivalent to $A \mid (\gnot A) B$, where $\gnot$ is the negation operator of PEGs. It is the $\gnot$ operator that we cannot do away with: our simulation of scaffolding automaton uses the $\gnot$ operator both for detecting the end of the input and for detecting the absence of a path in the scaffold. Interestingly, it is possible to modify the above construction to remove the second use case, by adding an extra family of non-terminal symbols $\mathsf{NoPath}(p)$, that accepts the input exactly when $p$ is not a valid path starting at that position. The result of this is that any parsing expression grammar may be replaced by a grammar where the operators appearing in parsing expressions are $\gand$, $|$, and the special symbol $\mathsf{EndOfInput}$, which accepts only at the end of the input. Details are left to the reader.

\section{Applications}\label{sec:applications}

In this section we will use Theorem \ref{thm:peg-automata} to prove all of the remaining results mentioned in the abstract.

\subsection{``Universality''}\label{sec:universality}

\begin{theorem}\label{thm:universality}
  Let $f:\ZO^\ast\to\ZO^\ast$ be any computable function. Then there exists a computable function $g: \ZO^\ast \to \bbN$ such that the language
  \[
    L = \{ f(x) \wait^{\ell} x \mid x \in \ZO^\ast, \ell \ge g(x) \} \subseteq \{0,1,\wait\}^\ast
  \]
  has a parsing expression grammar.
\end{theorem}

\begin{proof}
  We describe a scaffolding automaton for the reverse language $L^r$, and then the result follows from Theorem \ref{thm:peg-automata}. The basic idea is to use the \emph{reverse and scan} trick. For this purpose, let $M$ be a one-tape Turing machine computing $f$.

  The automaton first reads the input $x^r$, copying the symbols  of $x^r$ to the labels of the corresponding nodes and adding an edge connecting each node to the previous one.  It then finds the first $\wait$ symbol; at this point it continues reading $\wait$ symbols, while successively labelling the corresponding nodes of the scaffold with the successive configurations of the Turing machine $M$ on input $x$. After this it checks that the input matches the output of $M$ on input $x$. So, if $c_i$ is the configuration of $M$ on input $x$ at time-step $i$, and $M$ runs for $t$ time steps on input $x$, then the labels, when seen from first to last, form the string: \newcommand{\ovrs}[2]{\mathrlap{\hspace{1pt}#1}{\phantom{#2}}}\begin{align*}
      \text{labels:} \qquad & x^r \#\# \; c_0 \# c_0^r \#\# \; c_1 \# c_1^r \#\#  \; c_2 \# c_2^r \#\# \; \cdots \; c_t \# c_t^r \#\# \; \underbrace{\# \ldots \#}\\
      \text{input:} \qquad & x^r\ovrs{\wait}{\#}\ovrs{\wait}{\#}\; \ovrs{\wait}{c_0} \ovrs{\wait}{\#} \ovrs{\wait}{c_0^r} \ovrs{\wait}{\#}\ovrs{\wait}{\#} \; \ovrs{\wait}{c_1} \ovrs{\wait}{\#} \ovrs{\wait}{c_1^r} \ovrs{\wait}{\#}\ovrs{\wait}{\#}  \; \ovrs{\wait}{c_2} \ovrs{\wait}{\#} \ovrs{\wait}{c_2^r} \ovrs{\wait}{\#}\ovrs{\wait}{\#} \; \cdots \; \ovrs{\wait}{c_t} \ovrs{\wait}{\#} \ovrs{\wait}{c_t^r} \ovrs{\wait}{\#}\ovrs{\wait}{\#} \; \ovrs{\hspace{6pt}f(x)^r}{\# \ldots \#}
    \end{align*}
   Here $\#$ is being used as a separator. Note that $\$$ is also being used as a separator, but the symbol $\$$ is part of the actual language being recognized, and the symbol $\#$ is part of the alphabet being used to label the scaffold. 

  One may verify that the above labelling can be produced by a scaffolding automaton, provided we choose a reasonable encoding for Turing machine configurations (and for this purpose the working alphabet can be as large as desired). For example, we may encode a configuration by the sequence of symbols on the tape, and the position of the tape head will be additionally marked with some (finite) information containing the current state of the computation. With such an encoding, the scaffolding automaton can, for each $i$, produce the labels in the sequence $c_{i+1}$, provided that when reaching the first symbol of $c_{i+1}$, the top of the scaffold has an edge pointing to the last symbol of $c_i^r$ (which is easy to ensure), and that each node in the scaffold has an edge to the previous node; then the labelling $c_{i+1}$ is produced one symbol at a time by scanning $c_i^r$ starting with its last symbol, and  producing the symbols of $c_{i+1}$ according to the transition function of $M$. Similarly, for each $i$, one may produce the labels in the sequence $c_i^r$, provided that when reaching the first symbol of $c_i^r$, the top of the scaffold has an edge pointing to the last symbol of $c_i$; then the labelling $c_i^r$ is produced by copying one symbol at a time.

  The scaffolding automaton finally accepts if the last $\wait$ symbol corresponds exactly to the last position of the (reversal of) last configuration of the computation of $M$ on $x$, and the last $\wait$ symbol is followed by the string $y$ which is the reverse of the output written on the tape, in that final configuration; i.e.~if it is followed by $f(x)^r$.
\end{proof}

\bsni
We may now show that the recognition procedure underlying parsing expression grammars is complete for polynomial time, under logspace reductions. This was previously unknown, and stands in contrast with context-free grammars.
In the case of context-free grammars, we may define the complexity class $\mathsf{LOGCFL}$, to be the class of languages which are reducible to context-free languages under logspace reductions. It may be proven that this is exactly the class of languages decidable by log-depth Boolean circuits where the OR gates have arbitrary fan-in, and the AND gates have fan-in $2$ \cite[see][p.~137]{johnson1990catalog}. In particular, $\mathsf{LOGCFL}$ is a sub-class of $\NC_2$, which is believed to be strictly contained in $\PTIME$.

In contrast, if we were to define an analogous complexity class $\mathsf{LOGPEG}$, containing those languages that are reducible, via logspace reductions, to PEG-recognizable languages, it turns out that $\mathsf{LOGPEG} = \PTIME$. It is easy to see that $\mathsf{LOGPEG} \subseteq \PTIME$, since $\mathsf{PEG} \subseteq \PTIME$ and $\PTIME$ is closed under logspace reductions. The other direction follows as a corollary of Theorem \ref{thm:universality}.

\begin{corollary}
  There is a language $L\in\PEG$ which is complete for $\PTIME$ under logspace reductions.
\end{corollary}

\begin{proof}
  Notice in the proof of Theorem \ref{thm:universality} that the resulting function $g: \ZO^\ast \to \bbN$ grows quadratically in the running time of the Turing machine $M$. Now consider the function $f$ such that $f(x) = 1$ if $x$ encodes a triple $\langle N, 0^t, y \rangle$ where, in turn, $N$ encodes a Turing machine which accepts input $y$ in $t$ or fewer steps, and $t \ge |N| + |y|$. And let $f(x) = 0$ otherwise. Then, computing $f(x)$ is a problem which is complete for polynomial time under logspace reductions. There are machines for computing $f$ in time $O(t^2)$, and hence $g(\langle N, 0^t, y \rangle) = O(t^4) \le c\cdot t^4$ for some sufficiently large integer constant $c$. The language $L$ of Theorem \ref{thm:universality} is thus also complete for polynomial time under logspace reductions, since $f(\langle N, 0^t, y \rangle) = 1$ if and only if $1\wait^{c \cdot t^4} \langle N, 0^t, y \rangle \in L$, and the string $1\wait^{c \cdot t^4} \langle N, 0^t, y \rangle$ may be computed from $\langle N, 0^t, y \rangle$ in logarithmic space.
\end{proof}

\subsection{Impossibility of a Pumping Lemma}

\bsni
We may define a pumping lemma by the following:
\begin{definition}
  A \emph{pumping lemma for PEGs} is a total computable function $A$ such that, for every total\footnote{Although the totality of a given PEG is undecidable, the results of this section still hold if ``total'' is replaced by ``well-formed''. (Recall that well-formedness of PEGs is a decidable syntactic restriction which ensures totality. See remarks after Definition \ref{def:PEGclass}.) It should be understood, hence, that the impossibility of a pumping lemma is not a hidden consequence of the undecidability of totality.} PEG $G$, there exists a length $n_0$ such that for every string $x \in \cL(G)$ of size $|x| \ge n_0$, the output $y = A(G, x)$ is in $\cL(G)$ and has $|y| > |x|$.
\end{definition}

\noindent
Some explanation is required as to why this definition is the right one.
\begin{itemize}
\item The first observation we may make is that, to our knowledge, every pumping lemma proven thus far either already is of the above form (e.g. \cite{bar1964formal,yu1989pumping,amarilli2012proof}) or can be made to work in the above form with few modifications (e.g. considering resource-bounded Kolmogorov complexity in \cite{li1995new}). 
\item The second observation is that if $A$ is not required to be total, then the definition trivialises: there exists a pumping lemma for every recursively-enumerable language. Indeed given any Turing machine $M$ and input $x$, $A$ can simply dovetail on all $y$ larger than $x$ until it finds a larger $y$ accepted by $M$ (if no such $y$ is found, $M$ decides a finite language, and so the requirement on $A$ is trivially satisfied). 
\item We mention also that the definition is equivalent to one where $A$ is required to produce an infinite sequence $y\ps1, y\ps2, \ldots$ of strings of increasing size, which is what one typically sees in pumping lemmas.
\end{itemize}

\begin{theorem}\label{thm:no-pumping-lemma}
  There is no pumping lemma for PEGs.
\end{theorem}

We must show that any candidate computable function $A$ must fail on some grammar. Intuitively one may quickly realise, by way of Theorem \ref{thm:universality}, that the size of ``the next string'' in the language decided by a parsing expression grammar may well grow as high as any computable function of our choice. Hence given any candidate procedure $A$ meant to serve as a pumping lemma, we should be able to find a PEG language such that the gap between consecutive words grows faster than what the existence of $A$ would allow. The only difficulty in making this argument precise is that we wish to run algorithm $A$ on a PEG for the very same language we are trying to define. This is solved much the same way as in the proof of Kleene's second recursion theorem (see \cite{sipser2012introduction}, \S6.1): one shows that it is possible to construct a scaffolding automaton which has access to its own encoding.

  \begin{proof}
  For any scaffolding automaton $X$, let $\langle X \rangle$ be a binary encoding of $X$. Let $S \in \bbS(d, \Gamma)$ be a scaffold and $w \in \Gamma^n$. We say that \emph{$S$ sees $w$ written backwards} if, for every $\ell \in [n)$, following the first edge once and then the second edge $\ell$ times, from the top of $S$, will place us in a node labelled by $w_{n - \ell}$. Suppose we have a scaffolding automaton $C$, which accepts an input of the form $\wait^s \langle X' \rangle$, where $\langle X' \rangle$ in turn is the encoding of some scaffolding automaton $X'$. Let $\langle C \rangle$ be an encoding of $C$. We then define a scaffolding automaton $X_{\langle C \rangle}$, which recognises a language $\cL(X_{\langle C \rangle}) = \{y_1, y_2, \dots \}$, via the following procedure:
  \begin{itemize}
  \item $X_{\langle C \rangle}$ begins by checking that the input begins with $\langle C \rangle$, in such a way that after this check, the resulting scaffold sees $\langle C \rangle$ written backwards;
  \item $X_{\langle C \rangle}$ also maintains an edge from the current top node to the previous top node, at every step of the computation, and always copies the input into the labels of the scaffold, so it is not forgotten.
  \item Then $X_{\langle C \rangle}$ simulates a run of $C$ itself, which by assumption recognises a string of the form:
    \[
      \wait^s \langle X' \rangle
    \]
    An edge to the last symbol of $\langle X' \rangle$ is preserved by $X_{\langle C \rangle}$ throughout the rest of the computation (on every top node henceforth);
  \item Then $X_{\langle C \rangle}$ checks that the following input is the sequence $\# \mathsf{Start}\#$, and enters an accepting state at this point.
  \item The scaffold now sees the string $y_1 = \langle C \rangle \wait^s \langle X' \rangle \# \mathsf{Start} \#$ backwards.
  \item    Then for each $j = 1, 2, \ldots$, the automaton repeatedly:
    \begin{itemize}
    \item Simulates the computation of $A(G_{\langle X'\rangle}, y_j^r)$, in order to recognise an input of the form $\wait^{a_j} A(G_{\langle X'\rangle}, y_j^r) \#$, where $G_{\langle X'\rangle}$ is the grammar recognising the reverse of the language decided by $X'$. The grammar $G_{\langle X'\rangle}$ is (constructively) given by Theorem \ref{thm:peg-automata}, and the automaton can recognise an input of this form by way of Theorem \ref{thm:universality}. Here we require that $A$ is total.
    \item After scanning this input (while copying it into the labels of the scaffold), the automaton enters an accepting state.
    \item The scaffold now sees backwards:
      \[
        y_{j+1} = y_j \wait^{a_j} A(G_{\langle X' \rangle}, y_{j}^r)\#.
      \]
  \end{itemize}
  \end{itemize}

  \bsni
  Let $B$ be the scaffolding automaton which, under the assumption that the top of the scaffold sees an encoding $\langle C \rangle$ written backwards, accepts a string of the form
  \[
    \wait^b \langle X_{\langle C \rangle} \rangle.
  \]
  
  Such a scaffolding automaton $B$ exists, by Theorem \ref{thm:universality}. Let $\langle B \rangle$ be the code for the above scaffolding automaton. Then let us consider the scaffolding automaton $X_{\langle B \rangle}$, which accepts $y_1, y_2, \ldots$ --- this sequence is infinite by our assumption that $A$ is total. Note that setting $C = B$ satisfies the assumption that $X_{\langle C \rangle}$ makes on $C$. The string $\langle X' \rangle$ recognised during execution of $X_{\langle B \rangle}$ is exactly $\langle X_{\langle B \rangle} \rangle$. Hence $G_{\langle X'\rangle} = G_{\langle X_{\langle B \rangle} \rangle}$ is a parsing expression grammar deciding the same language as $X_{\langle B \rangle}$, in reverse. i.e.~$G_{\langle X_{\langle B \rangle}\rangle}$ recognises the strings $y_1^r, y_2^r, \ldots$. Now let $n_0$ be an arbitrary natural number, and consider $y_{n_0}$; clearly $|y_{n_0}| \ge n_0$; and yet the smallest string larger than $y_{n_0}$ which is accepted by $X_{\langle B \rangle}$ is $y_{n_0+1} = y_{n_0} \wait^{a_{n_0}} A(G_{\langle X_{\langle B\rangle} \rangle}, y_{{n_0}}^r)\#$ --- but its size is strictly greater than $A(G_{\langle X_{\langle B\rangle} \rangle}, y_{n_0}^r)$, and so is the size of $y_{n_0 + k}$ for any natural $k > 1$; hence $A$ must fail on the grammar $G_{\langle X_{\langle B \rangle}\rangle}$. 
\end{proof}

\subsection{PEGs vs.~Online Turing Machines}\label{sec:peg-vs-online}

Because scaffolding automata are machines which read a single input symbol at a time, and which do only a constant number of operations per symbol read, they can be thought of as a \emph{real-time} computational model. This led us to conjecture that the reverse of any language in $\PEG$ could be recognised by a real-time Turing machine.
However this conjecture turns out to be demonstrably false. 

Let us begin by the following definition:

\begin{definition}
  An \emph{online Turing machine} is a Turing machine where the head of the input tape can only move in one direction. At the beginning of the computation, an input $x \in \Sigma^\ast$ is written on the the input tape, and the head of the input tape sits over the leftmost symbol of $x$, and every time the tape head is moved to the right, we say that \emph{another symbol from the input was read}. For convenience, an additional auxiliary tape is provided where the input size $|x|$ is given in binary.\footnote{So that one will not think that the lower-bounds we are about to prove result, somehow, from the fact that the machine does not know the input size. Indeed the reason why the lower-bound holds is more profound. We may even fill the auxiliary tape with any content we please (as a function of $n$), i.e.~the lower-bounds here proven will hold even in the presence of non-uniform advice.}

  The class $\Online(t(n))$ is the class of languages $X \subseteq \Sigma^\ast$ which can be decided by an online Turing machine $M$, in the following way. If $x \in \Sigma^n$, then $M(x)$ accepts if $x \in X$ and rejects otherwise, and furthermore, the computation $M(x)$ does at most $t(n)$ steps between each input symbol read.
\end{definition}

\bsni
This section is devoted to proving the following:
\begin{theorem}\label{thm:non-real-time}
  There exists a language $L \in \PEG$ such that neither $L$ nor $L^r$ is in $\Online(t(n))$, for any $t(n) = o(n/(\log n)^2)$.
\end{theorem}

\bsni
The proof of this theorem uses the method of Rosenberg (see \cite{rosenberg1967real}, \S4.1), for proving lower-bounds against online Turing machines. We will explain it here for completeness.
\begin{definition}
  Let $L \subseteq \Sigma^\ast$ and $\ell, m \in \bbN$. We then say that two strings $y_1, y_2 \in \Sigma^\ell$ are $(L,\ell, m)$-equivalent, which we write $y_1 \equiv_L^{\ell,m} y_2$, if
  \[
    \forall x \in \Sigma^m (y_1 \cdot x \in L \iff y_2 \cdot x \in L)
  \]
  We may then define the sets $\cE_{L}(\ell, m) = \Sigma^\ast\slash\equiv_L^{\ell,m}$ of $(L, \ell, m)$-equivalence classes. To each $L \subseteq \Sigma^\ast$, then, corresponds a function $E_L:\bbN\times\bbN \to\bbN$ giving the number of $(L,\ell, m)$-equivalence classes:
  \[
    E_L(\ell, m) = |\cE_{L}(\ell, m)|
  \]
\end{definition}

\bsni
The framework of Rosenberg then rests on the following crucial observation:
\begin{theorem}[\cite{hartmanis1965computational}] \label{thm:not-online-method}
  If $L \in \Online(t(n))$, then $E_L(\ell, m) \le 2^{O(m \cdot t(\ell + m))}$.
\end{theorem}

\begin{proof}
  Let $M$ be an online Turing machine that decides whether $z \in L$ by making $\le t(|z|)$ computation steps per symbol. Let $y \cdot x \in \Sigma^n$, where $y \in \Sigma^\ell$, $x \in \Sigma^m$ and $n = \ell + m$. Consider the configuration $C$ of the computation $M(y \cdot x)$, after $M$ has read all the $\ell$ symbols in $y$ and done whichever computation it does on them, and precisely before it reads the first symbol of $x$. As $M$ then proceeds to read the $m$ symbols of $x$, it can only do $t(n)$ steps per symbol; and thus if one would describe the configuration $C$ partially, by giving only the state of the finite control, the position of the tape heads, and the contents of the tape heads at a distance of $\le m \cdot t(n)$ from the position of the tape heads, then one can simulate the entire computation to its very end.

  But since there are only $2^{O(m \cdot t(n))}$-many such possible partial descriptions, then this behaviour can only proceed in so-many different ways.
\end{proof}

\bsni
As a warm-up, we begin by showing the following easy result:
\begin{theorem}
  There is a language $K \subseteq \{0,1,\#\}^\ast$ in $\PEG$, such that $E_{K}((D+1)\cdot 2^D, D) \ge 2^{2^D}$ for all $D \in \bbN$. Hence $K \notin \Online(t(n))$, for any $t(n) = o(n/(\log n)^2)$.
\end{theorem}

\begin{proof}
  Consider the language:
  \[
    K^r = \{ x \# w_1 \# w_2 \# \cdots \# w_N \mid x \in \ZO^\ast, \forall i\; w_i \in \ZO^\ast, \exists i \; x^r = w_i \}.
  \]
  A scaffolding automaton can easily decide $K^r$ by maintaining an edge pointing to the last symbol of $x$, and then for each $w_i$ which it sees, scanning $x$ in reverse and comparing it with $w_i$. Hence $K \in \PEG$ by Theorem \ref{thm:peg-automata}.

  But looking carefully at $K = \{ w_1 \# w_2 \# \cdots \# w_N \# x \mid \exists i \; x^r = w_i \}$, one sees that if we have $N = 2^D$ strings $w_i$ each of length $D$, then the suffixes $x$ that cause acceptance are exactly those $x$'s in the set $\{w_1, \ldots, w_N\}$, and there are $2^N - 1 = 2^{2^D} - 1$ such sets. The empty set may be obtained by a malformed prefix, where none of the $w_i$ has length $D$, but their concatenation, with $\#$ as a separator, still has length $(D+1) 2^D$. Hence $E_K((D+1) 2^D, D) = 2^{2^D}$.

  Now, if $K$ were in $\Online(t(n))$ for $t(n) = o(n / (\log n)^2)$, by Theorem \ref{thm:not-online-method} we would have $E_K((D+1) 2^D, D) \le 2^{D \cdot t((D+1) 2^D)} = 2^{o(2^D)}$, a contradiction.
\end{proof}

\bsni
Now we will show the following:
\begin{theorem}\label{thm:not-online-hard-side}
  There is a language $H\subseteq\{0,1,\#\}^\ast$, decidable by a scaffolding automaton, such that $E_{H}(O(D \cdot 2^D), D ) \ge 2^{2^D}$.

  Hence $H \notin \Online(t(n))$, for any $t(n) = o(n/(\log n)^2)$.
\end{theorem}

The proof of this theorem is significantly more involved, and uses the \emph{reverse and scan} trick we have seen before. So let us first observe that from $K$ and $H$ we may obtain the language $L$ promised by Theorem \ref{thm:non-real-time}. Let $L^r = \{ 0 x 0 \mid x \in K^r \} \cup \{ 1 x 1 \mid x \in H \}$. It is easy to see that $L^r$ has a scaffolding automaton, since $K^r$ and $H$ both do, and so $L \in \PEG$. But an online Turing machine for deciding $L^r$ can be easily converted into an online Turing machine for deciding $H$, and an online Turing machine for deciding $L$ can be converted into an online Turing machine for deciding $K$. Hence neither $L^r$ nor $L$ are in $\Online(t(n))$ for any $t(n) = o(n/(\log n)^2)$.

\begin{proof}[Proof of Theorem \ref{thm:not-online-hard-side}]
  Given $(d, \Sigma)$-scaffold $G = (V, E, L)$ (as per Definition \ref{def:scaffold}) with $d \ge 2$, and a node $v$ of $G$, let us define the map $\Path_{G,v}: \ZO^\ast \to V \cup \{\varnothing\}$, so that $\Path_{G,v}(x_1 \cdots x_n) = v'$ if the sequence of bits $x_1 \cdots x_n \in \ZO^n$ is a valid path from $v$ to $v'$ in $G$ (as per Definition \ref{def:path}). If we repeat Definition \ref{def:scaffold} here, for explicitness, we get that $\Path_{G,v}$ is given inductively by
  \begin{itemize}
  \item $\Path_{G,v}(\lambda) = v$; and
  \item if $\Path_{G,v}(x_1 \cdots x_n) = w \neq \varnothing$ and $e_w$ is the edge list corresponding to node $w$ of $G$, then $\Path_{G,v}(x_1 \cdots x_n x_{n+1}) = e_w(x_{n+1})$; and 
  \item if $\Path_{G,v}(x_1 \cdots x_n) = \varnothing$, then $\Path_{G,v}(x_1 \cdots x_n x_{n+1}) = \varnothing$ also.
  \end{itemize}

  \medskip
  Then let us define the \emph{binary-depth} of $G$ with respect to $v$, $\BinDepth_G(v)$, to be the largest $D \in \bbN$ such that $\Path_{G,v}$ is ``total'' and injective on $\ZO^{\le D}$, i.e.~$\varnothing \notin \Path_{G,v}(\ZO^{\le D})$, and $|\Path_{G,v}(\ZO^{\le D})| = 2^{D+1}-1$. Intuitively explained: when recursively following the first two edges of $v$ and its descendants, we will find a complete binary tree of depth $D$. Note that, although in general, in a scaffold, we can have two distinct paths leading to the same node, our notion of binary-depth requires that all $2^{D+1}-1$ different paths in $\ZO^{\le D}$ lead to distinct nodes of $G$. If $\BinDepth_G(v) \ge D$, we will write $\BinTree_{G}(v, D)$ to denote the complete binary tree of depth $D$, rooted at $v$, obtained by recursively following the first two edges until depth $D$ is reached.

  \medskip
  A scaffolding automaton constructs a scaffold as it processes each new input symbol. We will devise a scaffolding automaton $\cA$ as follows. When $\cA$ is given any binary string $y \in \ZO^\ell$, with
  \begin{equation}
    \label{eq:ell}
    \ell = D + 1 + \sum_{n=0}^{2^{D+1} - 1}( 2 |n|_2 + 2 ) = O(D \cdot 2^D)
  \end{equation}
  where $|n|_2$ is the size of the smallest binary representation of the number $n$, then the resulting scaffold will have binary-depth $\ge D$, with respect to the first child of the top node. Formally said, the computation $\cA(y) = ((q_0, S_0), \ldots, (q_\ell, S_\ell))$ constructs the scaffold $S_\ell = ([\ell], E_\ell, L_\ell)$ having $\BinDepth_{S_\ell}(e_\ell(0)) \ge D$.

    \medskip
  Before showing how this is done let us show why it is enough. The language $H$ will be decided by a scaffolding automaton $\cA'$, in the following way: as long as $\cA'$ only sees $0$s and $1$s, it will run the algorithm of the automaton $\cA$. Besides the labels which $\cA$ places at each node, we also copy the corresponding input bit into that node, i.e.~the working alphabet of $\cA'$ will be the product of the working alphabet of $\cA$ with $\ZO$. Then we see our first separator symbol $\#$, and we stop running $\cA$. Let us call $z$ to the part of the input which precedes the separator symbol. After the separator, we expect to see a string $p \in \ZO^\ast$, and we interpret $p$ as if it were a path down the tree which is embedded in the scaffold. As we read the symbols of $p$, we thus maintain some edge following down this path. In this way we will traverse some bit positions of $z$, and we can see which bits of $z$ appear in these positions, since we have copied the bits of $z$ into the labels; then, whenever $z$ has a $1$ at such a position, we enter an accepting state, and whenever $z$ has a $0$, we enter a rejecting state.

  When $|z| = \ell$ as above, we have a full binary tree of depth $D$, and thus the strings $p \in \ZO^D$ will point to $2^D$ different positions of $z$. These positions are distinct (as required by the definition of binary-depth). Thus there are $2^{2^D}$ ways of filling such positions with bits. Each such way of filling these positions will give a different $(H, \ell, D)$-equivalence class. Hence
  \[
    E_H(\ell, D) \ge 2^{2^D}.
  \]

  \medskip
  Now to construct $\cA$. The base of the method is similar to how we built a scaffolding automaton for the counting language, in Section \ref{sec:expressive-power}. The scaffold constructed by $\cA$ will be labelled by the sequence
  \[
    (0)_2^r \circ (0)_2 \# \; (1)_2^r \circ (1)_2 \# \; \cdots \; (n-1)_2^r \circ (n-1)_2 \# \;  (n)_2^r \circ (n)_2 \# \cdots
  \]
  where for each natural number $k$, $(k)_2$ is its binary representation, and $(k)_2^r$ is the reverse of its binary representation. The characters $\#$ and $\circ$ are being used as separators, so $\#$ is called the \emph{outer separator}, and $\circ$ the \emph{inner separator}. It may be worthwhile to actually write it down:
  \[
    0\circ 0\#\; 1\circ 1\#\; 01 \circ 10 \# \; 11 \circ 11 \#\; 001 \circ 100 \#\; 101 \circ 101 \# \; 011 \circ 110 \# \; \ldots.
  \]
  It is not hard to see that such a labelling can be obtained by a scaffolding automaton: the automaton can copy what is before each inner separator symbol $\circ$ to appear after it in reverse, and then, after writing an outer separator symbol $\#$, it can scan the binary representation of the number $n$, appearing before the $\#$, from the lowest to highest-order bit, and apply the usual algorithm for incrementing a binary number by $1$, thus writing down the binary representation of $n+1$ in reverse.
  The nodes of the scaffold are thus divided into blocks, and the $n$-th block is of the form $(n)_2^r \circ (n)_2 \#$.

  We must now explain how the edges of the tree are added to the scaffold. The invariant we would like to preserve at the $n$-th block, is the following. Suppose $x_k \cdots x_1 = (n)_2$ is the binary representation of $n$, so that the $n$-th block is labelled by
  \[
    x_1 \cdots x_k \circ x_k \cdots x_1 \#
  \]
  Let $v_1 \cdots v_k\; r \; v_k' \cdots v_1'\; s$ be the nodes of the scaffold that get the labels above, i.e., the nodes of the scaffold corresponding to the $n$-th block. Then we would like to maintain the following property:
  
  \begin{invariant} It will always hold, on every block:
    \begin{itemize}
    \item If $x_i = 1$ for some $i \in \{2, \ldots, k\}$, then we will have $e_{v_i}(0) = e_{v'_i}(0)$, and $\BinDepth(e_{v_i}(0)) \ge i-1$. 
    \item Furthermore, for distinct $i, j \in \{2, \ldots, k\}$ with $x_i = x_j = 1$, the trees $\BinTree(e_{v_i}(0), i-1)$ and $\BinTree(e_{v_j}(0), j-1)$ are node-disjoint.
  \end{itemize}

  \end{invariant}
  I.e., one should think that if $x_i = 1$, the first edge leaving $v_i$ and $v_i'$ points to the root of the same full binary tree of depth $i - 1$. And that the two trees corresponding to different $v_i$ and $v_j$ share no node.

  For simplicity, let us momentarily ignore the ``Furthermore'' part of the invariant, and later argue that it will be upheld.

  Now suppose that this invariant holds for the $n$-th block, let us show how the algorithm needs to behave in order to make it hold for the $(n+1)$-th block. Suppose, for simplicity, that $n$ and $n+1$ are both $k$-bit numbers (the case when $n$ is $k$-bits and $n+1$ is $k+1$ bits is similar). Let $x_k \cdots x_1 = (n)_2$ and $y_k \cdots y_1 = (n+1)_2$ be their binary representations. The algorithm constructs the first half of the $(n+1)$-th block by scanning backwards the second half of the $n$-th block.

  So, suppose that the second half of the $n$-th block has nodes $v_k' \cdots v_1'$, which are labelled $x_k \cdots x_1$, respectively. Let $s$ be the node which is labeled by the outer separator $\#$ between blocks $n$ and $n+1$. Suppose that the algorithm is about to add the nodes $w_1 \cdots w_k$ to the first half of the $(n+1)$-th block, and intends to write the labels $y_1 \cdots y_k$ into them. This is done by reading $x_k \cdots x_1$ backwards: when the algorithm writes the label $y_1$ into $w_1$, he has an edge pointing to $v_1'$ where he can read $x_1$, when he writes $y_2$ into $w_2$ he has an edge pointing to $v_2'$, which is labelled by $x_2$, and so on. Such ``backwards scanning'' is easy to do provided we maintain an edge at each node which points to the previous node. The algorithm will also maintain an edge pointing to $s$.

  When incrementing $x_1 = 0$, then we will have $y_1 = 1$, and so we must make sure that $e_{w_1}(0)$ has binary-depth $\ge 0$: this is easily ensured by letting $e_{w_1}(0) = s$, $e_{w_1}(1) = \varnothing$.

  When incrementing $x_1 = 1$, we will have $y_1 = 0$; in this case we set $e_{w_1}(0) = v_1'$, and also set $e_{w_1}(1) = s$. It now holds that $\BinDepth(w_1) \ge 1$, and we will use this as the base case of an induction on the length of the $1$-prefix of $x$. This is illustrated in the figure below, for block number $n = 39$, so that $(n)_2 = x_6 \ldots x_1 = 100111$. So suppose that $x_{i-1} = 1, \ldots, x_1 = 1$ are the labels of $v_{i-1}', \ldots, v_1'$, and that we have written $y_1 = 0, \ldots, y_{i-1} = 0$ as the labels of $w_1, \ldots, w_{i-1}$. We are about to add the node $w_i$ to the first half of the $(n+1)$-th block, using our pointer to $v_i'$ in the second half of the $n$-th block. Suppose by induction that $\BinDepth(w_{i-1}) \ge i - 1$. Now look at $x_i$. If we are not finished with the $1$-prefix of $x$, i.e. if $x_i = 1$, then we must set $y_i = 0$. Our invariant for the previous block tells us that $\BinDepth(e_{v_i'}(0)) = i-1$, and our induction hypothesis gives us $\BinDepth(w_{i-1}) = i - 1$. So we create the new top node $w_i$ with $e_{w_i}(0) = e_{v_i'}(0)$ and $e_{w_i}(1) = w_{i-1}$, so that $\BinDepth(w_i) = i$. This satisfies our induction hypothesis. This case pertains to nodes $w_2$ and $w_3$ of the figure below. If we have reached the point where the carry stops, i.e., if $x_i = 0$, then we will set $y_i = 1$, and for this we create the new top $w_i$ and set $e_{w_i}(0) = w_{i-1}$, $e_{w_i}(1) = \varnothing$. This satisfies our invariant for the first half of $(n+1)$-th block (there is no carry in this case). This case pertains to node $w_4$ of the figure below. Notice how $\BinDepth(w_3) = 3$, i.e., we have a complete binary tree of depth $3$ rooted at $w_3$, which we have drawn in thicker lines for emphasis.

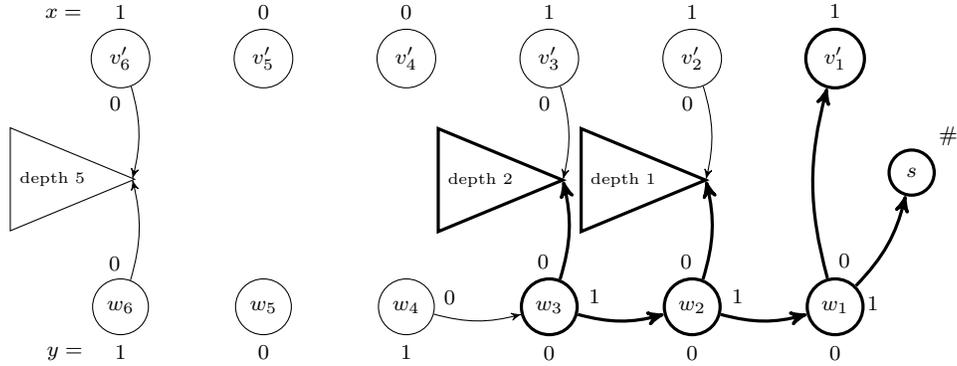
\begin{center}
\begin{tikzpicture}[node distance=1.9cm,>=stealth',bend angle=15,auto]\label{fig:real-time}

fill=black!20,minimum size=4mm]

                          \tikzstyle{snode}=[circle,draw=black,minimum size=6mm]
                          \tikzstyle{every path}=[->,every node/.style={font=\scriptsize}]
  
  \begin{scope}

    \node [snode,label={$1$}] (v'6) {$v'_6$};
    \node [] (x) [above left=0.13cm and 0.1cm of v'6]{$x = $};
    \node [isosceles triangle,draw] (v'6_0) [below left=1cm and 0.3cm of v'6] {\tiny depth $5$};
    \path 
    (v'6) edge [bend left=15] node [swap,at start] {$0$} (v'6_0.1);

    \node [snode,label={$0$}] (v'5) [right of=v'6] {$v'_5$};

    \node [snode,label={$0$}] (v'4) [right of=v'5] {$v'_4$};
    \node [snode,label={$1$}] (v'3) [right of=v'4] {$v'_3$};
    \node [isosceles triangle,draw,very thick] (v'3_0) [below left=1cm and 0.3cm of v'3] {\tiny depth $2$};
    \path 
    (v'3) edge [bend left=15] node [swap,at start] {$0$} (v'3_0.1);
    
    \node [snode,label={$1$}] (v'2) [right of=v'3] {$v'_2$};
    \node [isosceles triangle,draw,very thick] (v'2_0) [below left=1cm and 0.3cm of v'2] {\tiny depth $1$};
\path 
    (v'2) edge [bend left=15] node [swap,at start] {$0$} (v'2_0.1);

    \node [snode,label={$1$},very thick] (v'1) [right of=v'2] {$v'_1$};

    \node [snode,label=45:{$\#$},very thick] (s) [below right=1cm and 0.5cm
    of v'1] {$s$};
    
    \node [snode,label=-90:{$0$},very thick] (w1) [below=2.5cm of v'1] {$w_1$};
    \path
    (w1) edge [bend right=15,very thick] node [swap,at start] {$1$} (s)
    (w1) edge [bend left=15,very thick] node [swap,at start] {$0$} (v'1);
    
    \node [snode,label=-90:{$0$},very thick] (w2) [left of=w1] {$w_2$};
    \path
    (w2) edge [bend right=15,very thick] node [at start] {$1$} (w1)
    (w2) edge [bend right=15,very thick] node [at start] {$0$} (v'2_0.359);
    
    \node [snode,label=-90:{$0$},very thick] (w3) [left of=w2] {$w_3$};
    \path
    (w3) edge [bend right=15,very thick] node [at start] {$1$} (w2) 
    (w3) edge [bend right=15,very thick] node [at start] {$0$} (v'3_0.359);
    
    \node [snode,label=-90:{$1$}] (w4) [left of=w3] {$w_4$};
    \path 
    (w4) edge [bend right=15] node [at start] {$0$} (w3);
    
    \node [snode,label=-90:{$0$}] (w5) [left of=w4] {$w_5$};

    \node [snode,label=-90:{$1$}] (w6) [left of=w5] {$w_6$};
    \path 
    (w6) edge [bend right=15] node [at start] {$0$} (v'6_0.359);

    \node [] (y) [below left=0.13cm and 0.1cm of w6]{$y = $};
    
  \end{scope}
\end{tikzpicture}
\end{center}
  
  Once we find the first $x_i = 0$, we proceed by copying the remaining nodes and their edges; i.e.~we set $y_i = x_i$, $e_{w_i}(0) = e_{v_i'}(0)$, $e_{w_i}(1) = e_{v_i'}(1)$, until we find the inner separator $\circ$. After the inner separator $\circ$, we simply copy what we have done, i.e.~we set $y_i' = y_i$, $e_{w_i'}(0) = e_{w_i}(0)$, $e_{w_i'}(1) = e_{w_i}(1)$, until we find the outer separator $\#$.

  To keep things simple we have not considered the ``Furthermore'' part, so let us deal with it now. We have $y_1=0, \ldots, y_{i-1} = 0, y_i = 1$ for some $i$, which is the last point reached by the carry. Now notice that $\BinTree(w_{i-1}, i-1)$ (which is the tree under $w_3$ in the figure above) is made from ``fresh'' nodes, which did not previously belong to a tree, namely $w_1, \ldots, w_{i-1}$, $s$, and $v'_1$, together with the sub-trees $\BinTree(e_{v_j'}(0), j-1)$, for $1 < j < i$. These subtrees are, by the ``furthermore'' part of the invariant, disjoint from any sub-trees $\BinTree(e_{v_j'}(0), j-1)$ with $j \ge i$. Hence $\BinTree(w_{i-1}, i-1)$ will also be disjoint from $\BinTree(e_{w_j}(0), j-1)$, for $j \ge i$.

  \bsni
  The result of the above is that block number $2^{D+1}$ will have the labels
  \[
    0^D 1 \circ 1 0^D \#
  \]
  and if we let $v$ be the node which is labelled by the first $1$ appearing in this block, then we will have $\BinDepth(v) = D$. The expression (\ref{eq:ell}) for $\ell$ is simply the position of the input bit corresponding to the node $v$: we have $2^{D+1}-1$ many blocks before we reach the $2^{D+1}$-th block, and the $n$-th block has size $2|n|_2 + 2$; then we have the $D + 1$ symbols $0^D 1$, the last of which is at the position when the node $v$ is the top of the scaffold.
\end{proof}

\section*{Acknowledgement}

\thankstext

\section*{Bibliography}
\bibliographystyle{elsarticle-num}
\bibliography{pegcfl}

\end{document}